\documentclass[12pt]{amsart} 

\usepackage[utf8]{inputenc}
\usepackage[T1]{fontenc}	
\usepackage[french,english]{babel}

\usepackage{lmodern}			
\usepackage{newtxtext}

\usepackage[top=3cm, bottom=3cm, left=3.6cm, right=3.6cm]{geometry}

\usepackage[plainpages=false,colorlinks,linkcolor=bleuFonce,citecolor=rougeFonce,urlcolor=vertFonce,breaklinks]{hyperref}

\usepackage{graphicx}	

\usepackage{tikz}
\usetikzlibrary{patterns}

\usepackage{color}
\definecolor{vertFonce}{rgb}{0,0.5,0}
\definecolor{numLignes}{rgb}{0.17,0.57,0.7}	
\definecolor{gris}{rgb}{0.5,0.5,0.5}
\definecolor{grisFonce}{rgb}{0.2,0.2,0.2}
\definecolor{orange}{rgb}{1,0.65,0.31}		
\definecolor{orangeFonce}{rgb}{1,0.4,0}
\definecolor{bleuFonce}{rgb}{0,0,0.4}
\definecolor{rougeFonce}{rgb}{0.3,0,0}
\definecolor{rougeWord}{rgb}{0.5,0,0}
\definecolor{vertClair}{rgb}{0.8,1,0.8}
\definecolor{rougeClair}{rgb}{1,0.5,0.5}

\usepackage{pict2e}
\usepackage{multido}

\usepackage{amsfonts,amssymb,amsthm,amsmath} 
\usepackage{amsaddr}	
\usepackage{dsfont}		
\usepackage{mathrsfs}
\usepackage{yfonts}

\newtheorem{lem}{Lemma}[section]
\newtheorem{thm}{Theorem}
\newtheorem{cor}{Corollary}[section]
\newtheorem{prop}{Proposition}[section]

\newtheorem{remark}{Remark}[section]

\newenvironment{system*}{%
	\equation\nonumber\left\{\ \begin{aligned}
}{%
	\end{aligned} \right. \endequation%
}

\usepackage{stmaryrd}
\newcommand		{\subsetArrow}	{\mathrel{\ooalign{$\subset$\cr%
\hidewidth\raise-.087ex\hbox{$_\shortrightarrow\mkern-1.5mu$}\cr}}}
\newcommand		{\subsetarrow}	{\mathrel{\ooalign{$\subset$\cr%
\hidewidth\raise-1.45ex\hbox{$\vec{}\mkern6mu$}\cr}}}

\newcommand		{\N}		{\mathbb N}			
\newcommand		{\ZZ}		{\mathbb Z}			
\newcommand		{\Z}		{\ZZ}
\newcommand		{\RR}		{\mathbb R}			
\newcommand		{\R}		{\RR}
\newcommand		{\Rd}		{\R^d}
\newcommand		{\Rdd}		{\R^{2d}}

\renewcommand	{\SS}		{\mathds S}			

\newcommand		{\cH}		{\mathcal H}		
\renewcommand	{\L}		{\mathcal L}		
\newcommand		{\fS}		{\mathfrak S}		

\newcommand     {\cW}		{\mathcal W}		
\newcommand		{\cB}		{\mathcal B}

\newcommand		\sfT		{\mathsf T}			

\newcommand		{\lt}			{\left}				%
\newcommand		{\rt}			{\right}			%
\renewcommand	{\(}			{\lt(}
\renewcommand	{\)}			{\rt)}
\newcommand		{\set}[1]		{\lt\{#1\rt\}}
\newcommand		{\bangle}[1]	{\lt\langle #1\rt\rangle}
\newcommand		{\scalar}[2]	{\bangle{#1,#2}}

\newcommand		{\com}[1]		{\lt[{#1}\rt]}		

\newcommand		{\n}[1]			{\lt|{#1}\rt|}
\newcommand		{\nrm}[1]		{\lt\|{#1}\rt\|}
\newcommand		{\snrm}[1]		{\lVert #1\rVert}
\newcommand		{\bnrm}[1]		{\big\lVert #1\big\rVert}
\newcommand		{\Bnrm}[1]		{\Big\lVert #1\Big\rVert}
\newcommand		{\Nrm}[2]		{\nrm{#1}_{#2}}
\newcommand		{\sNrm}[2]		{\snrm{#1}_{#2}}
\newcommand		{\bNrm}[2]		{\bnrm{#1}_{#2}}
\newcommand		{\BNrm}[2]		{\Bnrm{#1}_{#2}}

\renewcommand		{\d}		{\mathop{}\!\mathrm{d}}		

\newcommand			{\Dx}		{\nabla_x}
\newcommand			{\Dv}		{\nabla_\xi}
\newcommand			{\conj}[1]	{\overline{#1}}		
\newcommand			{\Id}		{\mathrm{Id}}		

\DeclareMathOperator{\cF}		{\mathcal{F}}		
\DeclareMathOperator{\tr}		{Tr}				

\newcommand		{\F}[1]			{\cF\!\( #1 \)}		
\newcommand		{\Tr}[1]		{\tr\!\( #1 \)} 	

\newcommand		{\intd}			{\int_{\Rd}}
\newcommand		{\intdd}		{\int_{\Rdd}}
\newcommand		{\iintd}		{\iint_{\Rdd}}

\newcommand		{\sumj}			{\sum_{j\in J}}

\newcommand		{\jj}			{\mathrm{j}}	

\newcommand		{\cC}			{\mathcal{C}}

\usepackage{braket}
\newcommand		{\Inprod}[2]	{\Braket{#1 | #2}}

\newcommand		{\op}		{\boldsymbol{\rho}}	

\newcommand		{\opmu}		{\boldsymbol{\mu}}	

\newcommand		{\opnu}		{\boldsymbol{\nu}}	

\newcommand		{\tildop}		{\,\tilde{\!\op}}	
\newcommand		{\ttildop}		{\,\tilde{\tilde{\!\op}}}	

\newcommand		{\opp}		{\boldsymbol{p}}
\newcommand		{\opz}		{\boldsymbol{z}}

\newcommand		{\Dh}		{\boldsymbol{\nabla}}	
\newcommand		{\DDh}		{\boldsymbol{\Delta}}	
\newcommand		{\Dhx}[1]	{\Dh_{\!x} #1}			
\newcommand		{\Dhv}[1]	{\Dh_{\!\xi} #1}		
\newcommand		{\Dhvj}[1]	{\Dh_{\!\xi_\jj} #1}	

\newcommand		{\floor}[1]		{\lt\lfloor{#1}\rt\rfloor}

\newcommand		{\CS}	{\cC^\mathrm{S}}
\newcommand		{\CB}	{\cC^\mathrm{B}}
\newcommand		{\CSh}	{\widetilde{\cC}^\mathrm{\,S}}
\newcommand		{\CBh}	{\widetilde{\cC}^\mathrm{\,B}}

\title[\textsc{On Quantum Sobolev Inequalities}]{\Large On Quantum Sobolev Inequalities}

\author[\textsc{L.~Lafleche}]{\large\textsc{Laurent Lafleche}} 
\address{Institut Camille Jordan, UMR 5208 CNRS \\\& Université Claude Bernard Lyon 1, France}
\curraddr{\textsc{Unit\'e de Math\'ematiques Pures et Appliqu\'ees, \'Ecole Normale Supérieure de Lyon, allée d’Italie, 69364 Lyon, France}}
\email{laurent.lafleche@ens-lyon.fr}

\subjclass[2020]{46E35 $\cdot$ 81Q20 $\cdot$ 81S30 $\cdot$ 47A30 (81S07, 46N50).}
\keywords{Operator inequalities, trace inequalities, Sobolev inequalities, semiclassical approximation, uncertainty principle}

\begin{document}

\begin{abstract}\small\vspace{20pt}
	We investigate the quantum analogue of the classical Sobolev inequalities in the phase space, with the quantum Sobolev norms defined in terms of Schatten norms of commutators. These inequalities provide an uncertainty principle for the Wigner--Yanase skew information, and also lead to new bounds on the Schatten norms of the Weyl quantization in terms of its symbol. As an intermediate tool, we obtain the analogue of Hardy--Littlewood--Sobolev's inequalities for a semiclassical analogue of the convolution, and introduce quantum Besov spaces. Explicit estimates are obtained on the optimal constants.
\end{abstract}

\begingroup
\def\uppercasenonmath#1{} 
\let\MakeUppercase\relax 
\maketitle
\thispagestyle{empty} 
\endgroup

\bigskip

\renewcommand{\contentsname}{\centerline{Contents}}
\setcounter{tocdepth}{2}	
\tableofcontents

\section{Introduction}

	Semiclassical analysis allows to understand the similarities between quantum and classical mechanics, and more generally to obtain the leading behavior of highly oscillating waves. From a physical viewpoint, this corresponds to take units such that the constant $\hbar$ becomes negligible, or at least very small. However, error terms corresponding to high power of $\hbar$ often involve the knowledge of some uniform in $\hbar$ regularity.
	
	A typical object used to make the link between classical and quantum mechanics is the Weyl quantization, that associates to a function $f = f(x,\xi)$ of the phase space an operator
	\begin{equation}\label{eq:Weyl_def_0}
		\op_f := \iintd \widehat{f}(y,\xi) \,e^{2i\pi\(y\cdot x + \xi\cdot \opp\)}\d y \d\xi
	\end{equation}
	acting on $L^2(\Rd)$, by analogy with the Fourier inversion formula. Here, $(y,\xi)\in\Rdd$ are the phase space variables, $x$ is the operator of multiplication by the variable $x$, $\opp = -i\hbar\nabla$ is the quantum analogue of the momentum, and $\widehat{f}$ denotes the Fourier transform of $f$, defined by
	\begin{equation*}
		\F{f}\!(y,\xi) = \widehat{f}(y,\xi) = \iintd f(y',\xi')\, e^{-2i\pi\(y\cdot y'+\xi\cdot\xi'\)}\d y'\d \xi'. 
	\end{equation*}
	Equivalently, $\op_f$ is the operator with integral kernel
	\begin{equation*} 
		\op_f(x,y) = \intd e^{-2i\pi\(y-x\)\cdot\xi} \, f(\tfrac{x+y}{2},h\xi)\d\xi.
	\end{equation*}
	Another widely used object is the Wigner transform, that is the inverse operation: it associates to an operator $\op$ a function on the phase space
	\begin{equation}\label{eq:Wigner_def}
		f_{\op}(x,\xi) = \intd e^{-i\,y\cdot\xi/\hbar} \,\op(x+\tfrac{y}{2},x-\tfrac{y}{2})\d y.
	\end{equation}
	In general, these pointwise identities should be understood in the sense of distributions. One of the difficulties when trying to prove the propagation of regularity uniformly with respect to $\hbar$ of the Wigner transform of a solution for example of the Schrödinger or  Hartree equation is the fact that the Wigner transform is not positive and its Lebesgue norms $L^p(\Rdd)$ are not conserved except for $p=2$ (see e.g. \cite{lions_sur_1993}). Therefore, it is more convenient to work at the level of operators and to study conveniently scaled Schatten norms of these operators as a replacement of the Lebesgue norms, that we will denote by\footnote{Notice that the spaces of operators with finite $\L^p$ norm are just the standard Schatten spaces. If one wants to define a space taking into account the $\hbar$, then one should rather consider a space $\ell^\infty\L^p$ of sequences $(\op_\hbar)_{\hbar\in(0,1)}$ of operators indexed by $\hbar$ such that $\sup_{\hbar\in(0,1)} \Nrm{\op}{\L^p} < \infty$.}
	\begin{equation}\label{eq:def_norm}
		\Nrm{\op}{\L^p} = h^{\frac{d}{p}} \Nrm{\op}{p} = h^{\frac{d}{p}} \Tr{\n{\op}^p}^\frac{1}{p},
	\end{equation}
	where $\n{\op} = \sqrt{\op^*\op}$. The normalization implies in particular that $\|\op_f\|_{\L^2} = \Nrm{f}{L^2(\Rdd)}$ and $h^d\tr(\op_f) = \intdd f(x,\xi)\d x\d\xi$. For similar reasons, it is convenient to consider the quantum analogue of gradients in the phase space. The correspondence principle leads to define the quantum gradients by the following formulas
	\begin{equation}\label{eq:quantum_gradients}
		\Dhx \op := \com{\nabla,\op} \quad \text{ and } \quad \Dhv \op := \com{\frac{x}{i\hbar},\op}.
	\end{equation}
	As was noticed for instance in \cite{benedikter_mean-field_2014}, they can also be seen just as the Weyl quantization of the classical phase space gradients since
	\begin{equation}\label{eq:weyl_quantization_gradients}
		\op_{\Dx f} = \Dhx{\op_f} \quad \text{ and } \quad \op_{\Dv f} = \Dhv{\op_f}.
	\end{equation}	
	Motivated by the above quantum analogues of gradients and Lebesgue norms, we prove in this paper the corresponding quantum analogue of the classical Sobolev inequalities. In dimension $d=3$, it implies for example inequalities such as
	\begin{equation}\label{eq:case_p=2}
		\Nrm{\op}{\L^3}^2 \leq 2\,C^2 \Nrm{\Dhx\op}{\L^2} \Nrm{\Dhv\op}{\L^2}
	\end{equation}
	where $C = \CS_{1,2}+\frac{1}{2\sqrt{2\pi}} = \frac{1}{2\sqrt{2\pi}}\((\frac{20}{9})^{1/6}+1\)$ with $\CS_{1,2}$ the optimal constant in the classical Sobolev inequality in dimension $6$.
	
	One should not confuse the norms appearing on the right-hand side of the above inequality with norms of the form $\Nrm{(-\Delta)^{n/2}\op\,(-\Delta)^{n/2}}{\L^p}$ which are sometimes also used for operators, for example in works such as \cite{chadam_time-dependent_1976, brezzi_three-dimensional_1991, castella_$l^2$_1997, erdos_derivation_2001, lafleche_propagation_2019}. In the semiclassical approximation, these should rather be considered as the analogue of weighted Lebesgue norms or moments, since for sufficiently nice positive operators and for $p=1$, it holds
	\begin{equation*}
		\Nrm{\(-\Delta\)^{n/2}\op\(-\Delta\)^{n/2}}{\L^1} = \hbar^{-2n} \Nrm{\n{\opp}^n\op\n{\opp}^n}{\L^1} = \hbar^{-2n} \intdd f_{\op}(x,\xi) \n{\xi}^{2n} \d x\d\xi.
	\end{equation*}
	
	Inequality~\eqref{eq:case_p=2} can be seen as an uncertainty principle. If $\op\geq 0$ is an operator such that $\Tr{\op} = 1$, one can define its Wigner--Yanase skew information~\cite{wigner_information_1963} with respect to a self-adjoint operator $K$ by
	\begin{equation}\label{eq:skew}
		I_K(\op) = \frac{1}{2} \Tr{\n{\com{K,\sqrt{\op}}}^2}.
	\end{equation}
	The square root of this quantity is smaller than the standard variation $\sigma_K(\op)$, defined by $\sigma_K^2(\op) = h^d \tr(\op\n{K}^2) - (h^d \Tr{\op\,K})^2$, and reduces to it when $\op$ is a projection operator of rank one. Then in dimension $d=3$, Inequality~\eqref{eq:case_p=2} implies
	\begin{equation*}
		\sigma_x(\op)\,\sigma_{\opp}(\op) \geq \sqrt{I_x(\op) \, I_{\opp}(\op)} \geq \frac{\hbar}{8\pi C^2} \Nrm{\op}{3/2}.
	\end{equation*}
	The question of finding an analogue of the Heisenberg uncertainty inequalities for the skew information has been an active field of research. As was noticed in \cite{kosaki_matrix_2005} with a counterexample, it was first wrongly claimed in~\cite{luo_informational_2004} that the skew information satisfies the same uncertainty principle as the standard variation. Because of this fact, uncertainty principles for variants of the skew uncertainty where then investigated in several papers, such as~\cite{luo_heisenberg_2005, yanagi_wigneryanasedyson_2010, yang_generalized_2022}.

	The boundedness of the quantum gradients~\eqref{eq:quantum_gradients} in the above defined quantum Lebesgue norms~\eqref{eq:def_norm} independently of $\hbar$ is also at the core of several works dealing with semiclassical effective equations. It is used in~\cite{benedikter_mean-field_2014, benedikter_mean-field_2016-1, porta_mean_2017, golse_empirical_2019, chong_many-body_2021} in the context of the derivation of the Hartree equation in the mean-field regime, in~\cite{marcantoni_dynamics_2023} for the derivation of the Bogoliubov--de Gennes equations and in \cite{cardenas_effective_2023} for the case of Bose--Fermi mixtures, in~\cite{saffirio_semiclassical_2019, lafleche_strong_2023, chong_l2_2023} in the context of the derivation of the Vlasov equation from the Hartree--Fock equation in the semiclassical regime, in~\cite{golse_convergence_2021} in the context of the convergence of numerical schemes for quantum dynamics. Propagation of these quantities in the case of the Hartree and Hartree--Fock equation with Coulomb potential was proved in~\cite{chong_many-body_2021, chong_global--time_2022} where it was used to derive the Hartree--Fock and Vlasov equations with singular potentials in a combined mean-field and semiclassical limit. Other mixed norms of the quantum gradients were also considered in~\cite{porta_mean_2017, saffirio_semiclassical_2019}.
	
	In another context, a similar definition of quantum derivatives were also used in~\cite{bauer_self-adjointness_2023} in the study of Toeplitz operators. Quantum Sobolev spaces have also links with quantum generalizations of optimal transport distances. The Wigner--Yanase--Dyson information appears naturally when considering the self-distance of some of these generalizations \cite{de_palma_quantum_2021-1, lafleche_quantum_2023}, and some of these generalizations can be compared with quantum Sobolev norms of negative order \cite{golse_convergence_2021, golse_semiclassical_2021, lafleche_quantum_2023}.
	
	These semiclassical Schatten bounds on the quantum gradients are natural in the case when the system is in a state corresponding to a bounded distribution in the phase space. This type of bounds were proved to hold for certain pure states (that is projection operators) in~\cite{fournais_optimal_2020}, and in \cite{chong_semiclassical_2023} for thermal states. Other examples of more general states can be found in \cite{benedikter_mean-field_2014, lafleche_strong_2023, chong_many-body_2021, benedikter_effective_2022}. The regularity of projection operators in terms of the norms defined in this paper is studied in~\cite{lafleche_optimal_2023}.

\section{Notations, basic properties and results.}
	
\subsection{Classical Sobolev spaces}
	
	Denote by $z = (x,\xi) \in\Rdd$ the phase space variable. Then, in the classical setting, the homogeneous Sobolev spaces of order $1$ can be defined as the sets of phase space functions $f$ vanishing at infinity\footnote{In the sense that for any $\lambda>0$, the set $\set{x\in\Rdd : \n{f(x)}>\lambda}$ has finite measure.} such that the following norms are finite
	\begin{equation*}
		\Nrm{f}{\dot{W}^{1,p}(\Rdd)} = \Nrm{\nabla_z f}{L^p(\Rdd)}.
	\end{equation*}
	One can also define the corresponding non-homogeneous space endowed with the norm $\Nrm{f}{W^{1,p}(\Rdd)} = \Nrm{f}{L^p(\Rdd)} + \Nrm{\nabla_z f}{L^p(\Rdd)}$. By convention, we set $\dot{W}^{0,p} := L^p$. We will here mostly work on homogeneous spaces and not repeat these details.
	
	\subsubsection{Fractional Sobolev spaces} There is not a unique way to define the generalization to fractional order. We refer to \cite{stein_singular_1970, triebel_theory_1992, adams_sobolev_2003, mazya_sobolev_2011} for details on the extensions of Sobolev spaces to fractional order, with spaces such as the Besov spaces $B^s_{p,q}$ and the Triebel--Lizorkin spaces $F^s_{p,q}$. In the case $p=\infty$ and $s\in(0,1)$, the most commonly used spaces are the H\"older spaces $\dot{W}^{s,\infty} = \dot{B}^s_{\infty,\infty}$ defined by the norm
	\begin{equation*}
		\Nrm{f}{\dot{W}^{s,\infty}(\Rdd)} := \sup_{(z,w)\in\Rdd\times\Rdd\setminus\{z=w\}} \frac{\n{f(z)-f(w)}}{\n{z-w}^s}.
	\end{equation*}
	More generally, the most commonly called fractional Sobolev spaces are the Bessel--Sobolev spaces $\dot{H}^{s,p} = \dot{F}^s_{p,2}$ defined through the norm
	\begin{equation*}
		\Nrm{f}{\dot{H}^{s,p}(\Rdd)} := \Nrm{(-\Delta_z)^{s/2}f}{L^p(\Rdd)}
	\end{equation*}
	where $(-\Delta_z)^{s/2}$ denotes the fractional Laplacian and we will write $\dot{H}^{s} := \dot{H}^{s,2}$, and	the Sobolev--Slobodeckij spaces $\dot{W}^{s,p} = \dot{B}^s_{p,p}$ with norm
	\begin{equation}\label{eq:Sobolev_frac_def_0}
		\Nrm{f}{\dot{W}^{s,p}(\Rdd)}^p :=  \gamma_{s,p}\, \iint_{\Rdd\times\Rdd} \frac{\n{f(z)-f(w)}^p}{\n{z-w}^{2d+sp}}\d w\d z
	\end{equation}
	for $s\in(0,1)$, where $\gamma_{s,p}$ is chosen so that $\Nrm{f}{\dot{W}^{s,2}} = \Nrm{f}{\dot{H}^{s}}$, $\Nrm{f}{\dot{W}^{s,p}} \to \Nrm{f}{\dot{W}^{n,p}}$ when $s\to n$ with $n=0$ or $n=1$, and $\Nrm{f}{\dot{W}^{s,p}} \to \Nrm{f}{\dot{W}^{s,\infty}}$ when $p\to \infty$.
	
	\begin{remark}
		One can take for instance
		\begin{equation*}
			\gamma_{s,p} = \frac{p\n{\omega_{-2s}}}{4\,\omega_{2d+sp}} \(\frac{\pi\,\omega_{p+1}}{s^\frac{p-2}{2}}\)^s,
		\end{equation*}
		where $\omega_d = \lvert\SS^{d-1}\rvert = \frac{2\,\pi^{d/2}}{\Gamma(d/2)}$. Indeed it follows from the integral formula for the fractional Laplacian, for $s\in(0,1)$ and $z\in\Rdd$
		\begin{equation}\label{eq:lap_frac_integral_cl}
			(-\Delta_z)^\frac{s}{2} f(z) = c_{2d,s} \intdd \frac{f(w) - f(z)}{\n{w-z}^{2d+s}}\d w
		\end{equation}
		where $c_{2d,s} := \(2\pi\)^s \frac{\omega_{-s}}{\omega_{2d+s}}$, that $\gamma_{s,2} = \(2\pi\)^{2s} \frac{\n{\omega_{-2s}}}{2\,\omega_{2d+2s}}$. Moreover, it is not difficult to check using the Stirling formula that $\gamma_{s,p}^{1/p} \to 1$ when $p\to\infty$. On the other hand, it was proved in~\cite{bourgain_another_2001} (see also \cite[Equation~(44)]{brezis_how_2002}) that $\gamma_{s,p} \sim \frac{p\,\omega_{p+1}}{2\,\omega_{2d+p}} \(1-s\)$ when $s\to 1$ and it was proved in~\cite{mazya_bourgain_2002} that $\gamma_{s,p} \sim \frac{p\,s}{2\,\omega_{2d}}$ when $s\to 0$.
	\end{remark}
	
	\begin{remark}\label{rmk:bessel_vs_sobolev}
		In the case $s=1$ and $p\in(1,\infty)$ as well as in the case $p=2$, the spaces $\dot{W}^{s,p}$ and $\dot{H}^{s,p}$ coincide (see e.g.~\cite{calderon_lebesgue_1961}). Their norms are equal in the case when $p = 2$.
	\end{remark}
	
	\subsubsection{Classical Sobolev inequalities} One of the important feature of Sobolev spaces is the Sobolev inequality, originally proved in~\cite{sobolev_theorem_1938} and then generalized to many more spaces (see e.g.~\cite{lieb_analysis_2001, adams_sobolev_2003}). It tells that regularity implies additional integrability. More precisely, in the particular case of $s \in [0,1]$, when $1\leq p\leq q < \infty$ satisfy
	\begin{equation}\label{eq:exponent_sobolev}
		\frac{1}{p}-\frac{1}{q} = \frac{s}{2d}
	\end{equation}
	then there exists a constant $\CS_{s,p}$ such that for any $f\in \dot{W}^{s,p}(\Rdd)$
	\begin{equation}\label{eq:optimal_sobolev_cst}
		\Nrm{f}{L^q(\Rdd)} \leq \CS_{s,p} \Nrm{f}{\dot{W}^{s,p}(\Rdd)}.
	\end{equation}
	When $p > 1$, an analogue version holds for the Bessel--Sobolev spaces, i.e. there exists a constant $\CB_{s,p}$ such that for any $f\in \dot{H}^{s,p}(\Rdd)$
	\begin{equation}\label{eq:optimal_bessel_cst}
		\Nrm{f}{L^q(\Rdd)} \leq \CB_{s,p} \Nrm{f}{\dot{H}^{s,p}(\Rdd)}.
	\end{equation}
	As found by \cite{aubin_problemes_1976, talenti_best_1976}, when $s=1$, the optimal constant in Inequality~\eqref{eq:optimal_sobolev_cst} is given by
	\begin{equation*}
		\CS_{1,p} = \frac{1}{2d} \(\frac{q}{p'}\)^{1/p'} \(\frac{\omega_{4d/p}\,\omega_{4d/q'}}{\omega_{2d+2}\,\omega_{4d}}\)^{1/d}.
	\end{equation*}
	When $p=2$, the optimal constant was found in~\cite{lieb_sharp_1983} to be
	\begin{equation*}
		\CB_{s,2} = \CS_{s,2} = \frac{1}{\pi^s\,\omega_{2d+1}^{s/(2d)}}\(\frac{\omega_{2d+2s}}{\omega_{2d-2s}}\)^{1/2}.
	\end{equation*}
	To conclude this section, notice that replacing $f$ by $(-\Delta)^{t/2} f$ with $t \in (0,1)$ shows that more generally, if $0\leq t\leq s\leq 1$ and $1\leq p\leq q < \infty$ satisfy $\frac{1}{p}-\frac{1}{q} = \frac{s-t}{2d}$, then
	\begin{equation}\label{eq:regu_shift}
		\Nrm{f}{\dot{H}^{t,q}(\Rdd)} \leq \CB_{s-t,p} \Nrm{f}{\dot{H}^{s,p}(\Rdd)}.
	\end{equation}

\subsection{Quantum Sobolev spaces}\label{sec:quantum_Sobolev_spaces}
	
	The above considerations lead us to define the quantum analogue of Sobolev spaces in the phase space. For any compact operator $\op$, we define the homogeneous quantum Sobolev norms\footnote{This is in general a semi-norm, but gives a norm when restricted to the set of compact operators $\op$ such that $\Dh^n\op$ is compact.} of integer order by
	\begin{equation*}
		\Nrm{\op}{\dot{\cW}^{n,p}} := \Nrm{\Dh^n\op}{\L^p},
	\end{equation*}
	where we used the notation $\Dh = (\Dhx,\Dhv)$, so that $\n{\Dh\op}^2 = \n{\Dhx\op}^2 + \n{\Dhv\op}^2$. More details about vector valued operators can be found in Appendix~\ref{appendix:vec_operators}. In particular, when $n=1$, this norm is uniformly-in-$\hbar$ equivalent to the norm
	\begin{equation*}
		\Nrm{\Dhx\op}{\L^p} + \Nrm{\Dhv\op}{\L^p} = \Nrm{\com{\nabla,\op}}{\L^p} + \frac{1}{\hbar}\Nrm{\com{x,\op}}{\L^p}
	\end{equation*}
	used in~\cite{benedikter_hartree_2016, chong_many-body_2021}. As proved in Proposition~\ref{prop:commut_exponential}, this is also uniformly-in-$\hbar$ equivalent to the norm
	\begin{equation*}
		\Nrm{\Dhx\op}{\L^p} + \sup_{\xi\in\R^d\setminus\{0\}}\, \frac{1}{\hbar\n{\xi}} \Nrm{\com{e^{i\,\xi\cdot x}, \op}}{\L^p}
	\end{equation*}
	which is considered in~\cite{benedikter_mean-field_2014, benedikter_mean-field_2016-1, cardenas_effective_2023, marcantoni_dynamics_2023}.
	
	When $s\in(0,1)$, one can define the quantum analogue of the Gagliardo semi-norms $\dot{W}^{s,p}$ by setting for compact operators $\op$
	\begin{equation}\label{eq:Sobolev_frac_def}
		\Nrm{\op}{\dot{\cW}^{s,p}}^p :=  \gamma_{s,p}\, h^d\intdd \frac{\Tr{\n{\sfT_z\op - \op}^p}}{\n{z}^{2d+sp}}\d z.
	\end{equation}
	Here, $\sfT_z$ denotes the quantum phase space translation operator defined for $z_0 = (x_0,\xi_0)\in\Rdd$ by
	\begin{equation}\label{eq:translation_def}
		\sfT_{z_0}\op = e^{i\(\xi_0\cdot x -x_0\cdot\opp\)/\hbar} \op\, e^{i\(x_0\cdot\opp - \xi_0\cdot x\)/\hbar}.
	\end{equation}
	One can then define higher order Sobolev spaces by defining for $s\in\R_+\setminus\N$, $\Nrm{\op}{\dot{\cW}^{s,p}} = \Nrm{\Dh^n\op}{\dot{\cW}^{s-n,p}}$ with $n = \floor{s}$. In the particular case $p=\infty$, the limiting norm corresponding to \eqref{eq:Sobolev_frac_def} gives the quantum analogue of the H\"older norms, that is
	\begin{equation}\label{eq:Holder_norms}
		\Nrm{\op}{\dot{\cW}^{s,\infty}} := \sup_{z\in\Rdd\setminus\{0\}}\frac{\Nrm{\sfT_z\op - \op}{\L^\infty}}{\n{z}^{s}}.
	\end{equation}
	
	To define a quantum analogue of the Bessel-type Sobolev spaces, we start by defining the fractional Laplacian for operators. A convenient way is to go through the Wigner transform and Weyl quantization, as in Formula~\eqref{eq:weyl_quantization_gradients}. Hence we define
	\begin{equation}\label{eq:lap_frac_def}
		(-\DDh)^s\op := \op_{(-\Delta_z)^s f_{\op}}.
	\end{equation}
	One can then verify that $\DDh\op = \Dhx\cdot\Dhx\op + \Dhv\cdot\Dhv\op$ and get an analogue of the classical integral formulas such as Equation~\eqref{eq:lap_frac_integral_cl}. If $s\in(0,1)$, then
	\begin{equation}\label{eq:lap_frac_integral}
		(-\DDh)^\frac{s}{2}\op = c_{2d,s} \intdd \frac{\sfT_{z}\op - \op}{\n{z}^{2d+s}} \d z.
	\end{equation}
	We refer to Appendix~\ref{appendix:frac_lap} for more details on this quantum fractional Laplacian. Now, as a quantum analogue of the Bessel-type homogeneous Sobolev norms, one can define for any $s>-2d$,
	\begin{equation*}
		\Nrm{\op}{\dot{\cH}^{s,p}} := \Nrm{(-\DDh)^{s/2}\op}{\L^p}.
	\end{equation*}
	We will use the shortcut notation $\dot{\cH}^{s} = \dot{\cH}^{s,2}$. Since $\sNrm{\op_f}{\L^2} = \Nrm{f}{L^2(\Rdd)}$, one obtains from Remark~\ref{rmk:bessel_vs_sobolev} that for any compact operator $\op$,
	\begin{equation*}
		\Nrm{\op}{\dot{\cW}^{s,2}} = \Nrm{\op}{\dot{\cH}^s}.
	\end{equation*}
	To finish this section, let us comment about the case of a negative order of regularity. A general procedure to define negative Sobolev spaces is to proceed by duality and define for $s>0$ and an operator $\opmu$,
	\begin{equation*}
		\Nrm{\opmu}{\dot{\cW}^{-s,p}} := \sup_{\Nrm{\op}{\dot{\cW}^{s,p'}}\leq 1} h^d \Tr{\opmu\,\op}
	\end{equation*}
	where the supremum is taken over all compact operators $\op$ such that $\Nrm{\op}{\dot{\cW}^{s,p'}}\leq 1$. One can define similarly the non-homogeneous norm $\cW^{-s,p}$ by taking the supremum over all the compact operators satisfying $\Nrm{\op}{\cW^{s,p'}}\leq 1$ instead. Particular cases of these norms where used in \cite{golse_convergence_2021, golse_semiclassical_2021, lafleche_quantum_2023} where they are compared with optimal transport pseudo-norms. For example, it is proved in \cite[Proposition~B.2]{golse_semiclassical_2021} that
	\begin{equation*}
		\d(\op_1,\op_2) := \Nrm{\op_1-\op_2}{\cW^{-1,\infty}}
	\end{equation*}
	is a distance on the set of Hilbert--Schmidt operators which satisfies $\d(\op_1,\op_2) \leq 2^d \sup_{\Nrm{\varphi}{W^{1,\infty}\cap H^1}\leq 1} \intdd \(f_{\op_1}(z)-f_{\op_2}(z)\) \varphi(z)\d z$ and $\d(\op_1,\op_2) \leq \Nrm{\op_1-\op_2}{\L^\infty}$.

\subsection{Quantum Sobolev inequalities}

	We are now able to obtain the analogue of the classical Sobolev inequalities in terms of these Schatten-based Sobolev norms, everything being uniform in $\hbar$.

	\begin{thm}[\textbf{Gagliardo--Sobolev inequalities}]\label{thm:Gag_Sob}
		Let $s\in [0,1]$ and $1\leq p\leq q < \infty$ satisfying Identity~\eqref{eq:exponent_sobolev}. Then there exits a constant $\CSh_{s,p}$ independent of $\hbar$ such that for any compact operator $\op$,
		\begin{equation}\label{eq:Sobolev-Gagliardo}
			\Nrm{\op}{\L^q} \leq \CSh_{s,p} \Nrm{\op}{\dot{\cW}^{s,p}}.
		\end{equation}
		The optimal independent-of-$\hbar$ constant $\CSh_{s,p}$ satisfies
		\begin{align}\label{eq:Gag_Sob_cst_1}
			\CS_{s,p} \leq \CSh_{s,p} &\leq \CS_{s,p} + \frac{\omega_{2d}}{\omega_{2d+1}} & \text{ when } s=1
			\\\label{eq:Gag_Sob_cst_2}
			\CS_{s,p} \leq\CSh_{s,p} &\leq \CS_{s,p} + \tfrac{\theta_{s}}{\(8\pi\)^{s/2}} & \text{ when } p=2
			\\\label{eq:Gag_Sob_cst_frac}
			\CSh_{s,p} &\leq \CS_{s,p} + \frac{\(1-2^{\(s-1\)r}\)^\frac{1}{r}}{\gamma_{s,p}^{1/p}\(p'\)^{d+\frac{s}{2}}} \(\frac{\omega_{2d}}{\omega_{(2d+s)p'}}\)^\frac{1}{p'}  & \text{ when } s\in (0,1)
		\end{align}
		where $r = \max(p,p')$, $\CS_{s,p}$ is the constant of the corresponding classical Sobolev inequality~\eqref{eq:optimal_sobolev_cst} and $\theta_s = \sup_{r\geq 0} \frac{1-e^{-r}}{r^s}$ satisfies $\theta_s \leq 1$. The second inequality in Equation~\eqref{eq:Gag_Sob_cst_2} is also satisfied when $q=2$.
	\end{thm}
	
	In the case when $s=1$, applying the above theorem (together with Proposition~\ref{prop:schatten_norm_vec_valued_ineq}) to $\op_\lambda(x,y) = \lambda^d \op(\lambda\,x,\lambda\,y)$ and optimizing with respect to $\lambda$ implies the following result.
	\begin{cor}[\textbf{Sobolev inequalities}]
		If $1 \leq p\leq q<\infty$ satisfy $\frac{1}{2d} = \frac{1}{p} - \frac{1}{q}$ then for any compact operator $\op$,
		\begin{equation}\label{eq:Sobolev}
			\Nrm{\op}{\L^q} \leq 2^{1/r}\,\CSh_{1,p} \Nrm{\Dhx\op}{\L^p}^{1/2} \Nrm{\Dhv\op}{\L^p}^{1/2} \leq \CSh_{1,p} \(\Nrm{\Dhx\op}{\L^p}^r + \Nrm{\Dhv\op}{\L^p}^r\)^\frac{1}{r}
		\end{equation}
		where $r = \min(p,2)$.
	\end{cor}
	
	As claimed in the introduction, these inequalities provide an uncertainty principle for the skew information defined in Equation~\eqref{eq:skew}. Taking $p=2$ in the above inequalities applied to $\sqrt{\op}$ indeed yields the following.
	\begin{cor}[\textbf{Uncertainty for the Wigner--Yanase information}]
		For any compact operator $\op$ acting on $L^2(\Rd)$ with $d\geq 2$, it holds
		\begin{equation*}
			\sqrt{I_x(\op) \, I_{\opp}(\op)} \geq \frac{\hbar}{8\pi \,(\CSh_{1,2})^2} \Nrm{\op}{\frac{d}{d-1}}
		\end{equation*}
		where $\CS_{1,2} \leq\CSh_{1,2} \leq \CS_{1,2} + \tfrac{1}{\sqrt{8\pi}}$ with $(\CS_{1,2})^2 = \(4\pi\)^{-1/(2d)}\frac{\Gamma(d+1/2)^{1/d}}{d \(d-1\)\pi}$.
	\end{cor}
	
	It is interesting to notice that for $d$ large, $\hbar/(8\pi \,(\CS_{1,2})^2) \Nrm{\op}{\frac{d}{d-1}} \sim \frac{e}{4} \,\frac{d\,\hbar}{2} \Nrm{\op}{1}$, which has the same order in $d$ as lower bound in the classical Heisenberg inequality
	\begin{equation*}
		\sigma_x(\op)\,\sigma_{\opp}(\op) \geq \frac{d\,\hbar}{2} \Nrm{\op}{1}.
	\end{equation*}
	We conjecture that the same asymptotic behavior holds true for $\CSh_{1,2}$. 
	
	\begin{remark}[Operator norm of the Weyl quantization]
		In the case when $q=\infty$, the classical Sobolev embedding breaks, but a proof similar to the proof of Theorem~\ref{thm:Gag_Sob} yields the following interesting bound for the operator norm of the Weyl quantization~\eqref{eq:Weyl_def_0} in terms of its symbol
		\begin{equation*}
			\sNrm{\op_f}{\L^\infty} \leq \Nrm{f}{L^\infty(\Rdd)} + C \Nrm{f}{\dot{B}^d_{2,\infty}(\Rdd)},
		\end{equation*}
		where the classical homogeneous Besov norms can be defined by the formula
		\begin{equation*}
			\Nrm{f}{\dot B^s_{p,r}(\Rdd)} = \Nrm{\frac{\Nrm{f(\cdot-2z) - 2\, f(\cdot-z) + f}{L^p(\Rdd)}}{\n{z}^{s+2d/r}}}{L^r_z(\Rdd)}
		\end{equation*}
		when $s\in(0,1]$ and by $\Nrm{f}{\dot B^s_{p,r}} = \Nrm{\nabla^nf}{\dot B^\alpha_{p,r}}$ with $n = \floor{s}$ and $\alpha = s-n$ if $s> 1$ (see e.g. \cite{bergh_interpolation_1976}). This is a slight improvement, in terms of Besov spaces, with respect to the classical bound that follows from Formula~\eqref{eq:Weyl_def_0}
		\begin{equation*}
			\sNrm{\op_f}{\L^\infty} \leq \Nrm{\F{f}}{L^1(\Rdd)} \leq C \Nrm{f}{\dot{B}^d_{2,1}(\Rdd)}
		\end{equation*}
		since $\dot{B}^d_{2,1}(\Rdd) \subset L^\infty(\Rdd)\cap \dot{B}^d_{2,\infty}(\Rdd)$. This complements works such as~\cite{karlovich_algebras_2007} about the boundedness of operators with discontinuous symbol. From the above theorem with $p=2$ and $q<\infty$, one gets the following bound for the Schatten norm of the Weyl quantization
		\begin{equation*}
			\sNrm{\op_f}{\L^q} \leq \CSh_{s,2} \Nrm{f}{\dot{H}^s(\Rdd)} \qquad \text{ with } \frac{s}{2d} = \frac{1}{2} - \frac{1}{q}.
		\end{equation*}
		One can also compare the above results with what is obtained using the Calderòn--Vaillancourt theorem together with the fact that $\sNrm{\op_f}{\L^2} = \Nrm{f}{L^2(\Rdd)}$. If $n = \floor{d/2}+1$, by an improved version of the Calderòn--Vaillancourt theorem for the Weyl quantization proved by Boulkhemair~\cite[Theorem~1.2]{boulkhemair_l2_1999}, the following inequality holds \begin{equation*}
			\sNrm{\op_f}{\L^\infty} \leq C \Nrm{f}{W^{2n,\infty}(\R^{2d})}.
		\end{equation*}
		This latter inequality requires at least $d$ derivatives in $L^\infty(\Rdd)$ instead of $L^2(\Rdd)$, which is a stronger requirement locally but does not require decay at infinity.
	\end{remark}
	
	As in the classical case, there is an analogue of the quantum Sobolev inequality in the case of the Bessel--Sobolev spaces whenever $p > 1$. One can also estimate the constant in this case, however we do not give such estimates in the general case, as the proof given here uses interpolation theory and equivalent norms, which makes much more involved the computations of the constants. It is an interesting question to know if it is possible to get a more elementary proof.

	\begin{thm}[\textbf{Bessel--Sobolev inequalities}]\label{thm:Bes_Sob}
		Let $s\in [0,1]$ and $1< p\leq q < \infty$ satisfying~\eqref{eq:exponent_sobolev}. Then there exist a constant $\CBh_{s,p}$ such that for any compact operator~$\op$,
		\begin{equation}\label{eq:Sobolev-Bessel}
			\Nrm{\op}{\L^q} \leq \CBh_{s,p} \Nrm{\op}{\dot{\cH}^{s,p}}.
		\end{equation}
		When $p=2$ then $\CBh_{s,p} = \CSh_{s,p}$. When $q=2$ then $\CBh_{s,p} \leq \CS_{s,p} + \frac{\theta_s}{(8\pi)^{s/2}}$. More generally, when $p \leq 2 \leq q$, a bound on the constant is given in Proposition~\ref{prop:Bessel-sobolev_cst}.
	\end{thm}
	
	\begin{remark}
		For the Bessel--Sobolev spaces or the Sobolev spaces of integer regularity, the analogue of Inequality~\eqref{eq:regu_shift} holds. More generally, replacing $\op$ by $\Dh^n\op$ or $(-\DDh)^\alpha\op$, it is not difficult to obtain inequalities for higher order Sobolev spaces, and from one Sobolev space to another.
	\end{remark}
	
	Let us also indicate that in the above Sobolev inequalities, the relation~\eqref{eq:exponent_sobolev} always implies that $p < \frac{2d}{s}$. In the case when $p> \frac{2d}{s}$, then a similar proof as the proof of Theorem~\ref{thm:Gag_Sob} gives the analogue of the Morrey--Sobolev inequalities, which are the embedding of Sobolev spaces in H\"older spaces. In the quantum case, this gives the following inequalities in terms of the norms~\eqref{eq:Holder_norms}.
	\begin{thm}[\textbf{Morrey--Sobolev inequalities}]\label{thm:Mor_Sob}
		Let $0<\theta\leq s\leq 1$ and $1\leq p < \infty$ satisfy
		\begin{equation*}
			\frac{1}{p} = \frac{s-\theta}{2d}.
		\end{equation*}
		Then there exists a constant $\CSh_{s,p}$ independent of $\hbar$ such that for any compact operator~$\op$,
		\begin{equation}\label{eq:Sobolev-Morrey}
			\Nrm{\op}{\dot{\cW}^{\theta,\infty}} \leq \CSh_{s,p} \Nrm{\op}{\dot{\cW}^{s,p}}.
		\end{equation}
	\end{thm}

\subsection{Quantum Besov spaces}\label{sec:besov}

	More generally, one can also consider the quantum analogue of Besov spaces by defining for example for $s\in(0,2)$,
	\begin{equation}\label{eq:quantum_Besov}
		\Nrm{\op}{\dot{\cB}^s_{p,r}} :=  \BNrm{\frac{\sfT_{2z}\op - 2\,\sfT_z\op + \op}{\n{z}^{s+2d/r}}}{L^r_z(\L^p)}.
	\end{equation}
	They are in particular natural when considering the maximal regularity of projection operators~\cite{lafleche_optimal_2023}. Since $\sfT_z$ preserves the Schatten norm, one can replace $\sfT_{2z}\op+\op-2\sfT_z\op$ by $\sfT_z\op+\sfT_{-z}\op-2\op$. As in the classical case, one can find a Littlewood--Paley decomposition of these spaces (see Section~\ref{sec:Littlewood-Paley}), using the semiclassical convolution product presented in Section~\ref{sec:semiclassical_convolution}. This gives an easy generalization of these spaces to other orders, up to the equivalence of norms, and proves the ordering of Besov spaces with respect to the third index (Corollary~\ref{cor:besov_inclusions}), the comparison of Besov spaces of order~$0$ with Schatten norms (Corollary~\ref{cor:besov_vs_schatten}) and Sobolev's inequalities for Besov spaces (Corollary~\ref{cor:besov_Sobolev_ineq}). When $s\in(0,1)$, the norms $\dot{\cW}^{s,p}$ and $\dot{\cB}^s_{p,p}$ are equivalent independently of $\hbar$.
	
	In the particular case when $p=2$, due to the fact that the Wigner transform is an isomorphism from $\L^2$ to $L^2(\Rdd)$, these spaces are just the classical Besov spaces for the Wigner transform, i.e.
	\begin{equation*}
		\Nrm{\op}{\dot\cB^s_{2,r}} = \Nrm{f_{\op}}{\dot B^s_{2,r}(\Rdd)},
	\end{equation*}

	As in the classical case, Besov spaces are intermediate spaces between Sobolev spaces. In particular, the inequality in the next proposition is reminiscent of the possible definition of classical Besov spaces as real interpolation of Sobolev spaces.
	
	\begin{prop}[Interpolation inequality for Besov spaces]\label{prop:comparison_besov_sobolev}
		For any $(s,p,r)\in(0,1)\times[1,\infty]^2$, there exists a constant $C>0$ such that for any $\op\in\L^p$.
		\begin{equation*}
			\tfrac{1}{2}\Nrm{\op}{\dot{\cB}^s_{p,r}} \leq C_s \Nrm{\op}{\L^p}^{1-s} \Nrm{\op}{\dot{\cW}^{1,p}}^s.
		\end{equation*}
		The constant satisfies $C_s \leq 2^{1-s}\(\tfrac{\omega_{2d}}{s\(1-s\)r}\)^{1/r}$.
	\end{prop}
	
	By the same proof, taking $r=p$ and using the definition~\eqref{eq:Sobolev_frac_def} of the $\dot{\cW}^{s,p}$ norm, it yields
	\begin{equation*}
		\Nrm{\op}{\dot{\cW}^{s,p}} \leq C_{s,p} \Nrm{\op}{\L^p}^{1-s} \Nrm{\op}{\dot{\cW}^{1,p}}^s
	\end{equation*}
	where using the fact that $\n{\omega_{-2s}} = \tfrac{2s\(1-s\)}{\pi^s\Gamma(2-s)}$, one obtains that $C_{s,p}^p = \tfrac{2^{\(1-s\)p}\omega_{2d}}{s\(1-s\)p}\,\gamma_{s,p} \leq \frac{\omega_{2d}\,\omega_{p+1}^s}{s^{sp/2}\,\omega_{2d+sp}} $ is uniformly bounded with respect to $s\in[0,1]$.
	
\subsection{Equivalent norms for quantum Sobolev spaces}

\subsubsection{Besov norms and Sobolev norms}

	From the Definition~\eqref{eq:quantum_Besov}, it follows directly from the triangle inequality that $\gamma_{s,p}^{1/p} \Nrm{\op}{\dot{\cB}^s_{p,p}} \leq 2 \Nrm{\op}{\dot{\cW}^{s,p}}$. It is not difficult to get an inequality in the other direction using the well-known formula $\sfT_z-1 = \frac{1}{2}\(\sfT_{2z}-1\) - \frac{1}{2}\(\sfT_{2z}+1-2\sfT_z\)$. This yields
	\begin{equation}\label{eq:comparison_besov_sobolev_frac}
		\(2-2^s\)\Nrm{\op}{\dot{\cW}^{s,p}} \leq \gamma_{s,p}^{1/p} \Nrm{\op}{\dot{\cB}^s_{p,p}} \leq 2 \Nrm{\op}{\dot{\cW}^{s,p}}.
	\end{equation}
	When $s=1$, we can also apply the fundamental theorem of calculus to the Wigner transform $f_{\op}$ of $\op$ defined by Equation~\eqref{eq:Wigner_def}, and then take the Weyl quantization to get 
	\begin{equation}\label{eq:fundamental_theorem}
		\sfT_z\op - \op = z\cdot\int_0^1 \sfT_{\theta z} \Dh\op \d\theta
	\end{equation}
	%
	from which we deduce that
	\begin{equation}\label{eq:comparison_besov_sobolev}
		\Nrm{\op}{\dot{\cB}^1_{p,\infty}} \leq 2 \Nrm{\op}{\dot{\cW}^{1,p}}.
	\end{equation}
To get better constants in Sobolev inequalities when $s\to 1$, it is useful to obtain an inequality sharper than Inequality~\eqref{eq:comparison_besov_sobolev_frac}, in which the constant $\gamma_{s,p}$ converges to~$0$.
	\begin{prop}\label{prop:comparison_besov_sobolev_frac}
		Let $(s,p,q)\in (0,1)\times[2,\infty]^2$ and $\op\in\dot{\cW}^{s,p}$. Then $\op\in\dot{\cB}^s_{p,p}$ and
		\begin{equation}\label{eq:comparison_besov_sobolev_frac_1}
			\frac{(2^{p'} - 2^{sp'})^{1/p'}}{\gamma_{s,p}^{1/p}} \Nrm{\op}{\dot{\cW}^{s,p}} \leq \Nrm{\op}{\dot{\cB}^s_{p,p}} \leq \(\frac{2^p - 2^{sp}}{\gamma_{s,p}}\)^\frac{1}{p} \Nrm{\op}{\dot{\cW}^{s,p}}.
		\end{equation}
		When $p\in[1,2]$, the inequalities are reversed. In particular, when $p=2$, then
		\begin{equation}\label{eq:comparison_besov_sobolev_frac_2}
			\Nrm{\op}{\dot{\cH}^s}^2 = \Nrm{\op}{\dot{\cW}^{s,2}}^2 = \frac{\gamma_{s,2}}{4 - 4^s} \Nrm{\op}{\dot{\cB}^s_{2,2}}^2.
		\end{equation}
	\end{prop}
	
	\begin{remark}
		These constants are sharper when $s\to 1$. The constant on the right-hand side of Inequality~\eqref{eq:comparison_besov_sobolev_frac_1} converges to $\frac{2^{p+1}\ln(2)\,\omega_{2d+p}}{\omega_{p+1}}$ when $s\to 1$. In particular, when $p=2$, this yields
		\begin{align*} 
			\bNrm{(-\DDh)^\frac{1}{2}\op}{\L^2}^2 = \Nrm{\Dh\op}{\L^2}^2 = \frac{d}{\ln(4)\,\omega_{2d}} \Nrm{\op}{\dot{\cB}^1_{2,2}}^2.
		\end{align*}
	\end{remark}
	
\subsubsection{Commutators with complex exponentials}

	In \cite{benedikter_mean-field_2014, benedikter_mean-field_2016-1, cardenas_effective_2023, marcantoni_dynamics_2023}, instead of assumptions on the commutator $\Dhv\op = \com{\frac{x}{i\hbar},\op}$, one can find an assumption of the form
	\begin{equation*}
		\Nrm{\com{e^{ix\cdot \xi},\op}}{\L^p} \leq C \n{\xi} \hbar
	\end{equation*}
	for some constant $C$ independent of $\hbar$. As the next proposition shows, this is actually equivalent to the assumption of boundedness of $\Nrm{\Dhv\op}{\L^p}$ uniformly in $\hbar$.
	\begin{prop}\label{prop:commut_exponential}
		Let $p\in[1,\infty]$. Then for any $\op\in\cW^{1,p}$, it holds
		\begin{equation}\label{eq:commut_exponential}
			\frac{1}{\sqrt{d}} \Nrm{\Dhv\op}{\L^p} \leq \sup_{\xi\in\R^d\setminus\{0\}} \frac{\Nrm{\com{e^{i\,\xi\cdot x}, \op}}{\L^p}}{\hbar\n{\xi}} \leq \Nrm{\Dhv\op}{\L^p}.
		\end{equation}
		More precisely, for any $\jj\in\set{1,\dots,d}$, it holds
		\begin{equation}\label{eq:commut_exponential_j}
			\Nrm{\Dhvj{\op}}{\L^p} = \sup_{\xi_\jj\neq 0}\frac{\Nrm{\sfT_{\xi_\jj}\op - \op}{\L^p}}{\n{\xi_\jj}} = \sup_{\xi_\jj\neq 0}\frac{\Nrm{\com{e^{i\,\xi_\jj\,x_\jj}, \op}}{\L^p}}{\hbar\n{\xi_\jj}}.
		\end{equation}
	\end{prop}
	
\subsubsection{Boundedness of Riesz transforms in the quantum phase space}

	As indicated in Remark~\ref{rmk:bessel_vs_sobolev}, the Bessel potential spaces and the Sobolev spaces coincide in the case of an integer order of regularity. This can equivalently be stated in terms of the boundedness of the Riesz transform in Lebesgue spaces, and more generally the boundedness of integral operators or Fourier multipliers. In the non-commutative case, similar inequalities exist, see for example~\cite{caspers_schur_2015, junge_noncommutative_2018, conde-alonso_schur_2022} and the references therein. In particular, it follows from \cite[Theorem~A]{conde-alonso_schur_2022} that the Schur multiplier $(-\DDh_{\xi})^{-1/2}\,\Dh_\xi$ which acts on an operator $\op$ through the formula
	\begin{equation*}
		\((-\DDh_{\xi})^{-1/2}\,\Dh_{\xi}\op\)(x,y) = -i\,\frac{x-y}{\n{x-y}}\, \op(x,y)
	\end{equation*}
	is bounded on $\L^p$ for $1<p<\infty$. From this result, we deduce in particular that the quantities $\Nrm{\Dhv{\op}}{\L^p}$ and $\Nrm{(-\DDh_{\xi})^{1/2}\op}{\L^p}$ are equivalent semi-norms. By conjugation by the Fourier transform, the same result holds when exchanging $\xi$ by $x$ as in~\cite[Section~3.2]{chong_l2_2023}, that is $\Nrm{\Dhx{\op}}{\L^p}$ and $\Nrm{(-\DDh_x)^{1/2}\op}{\L^p}$ are also equivalent. In particular, there exists constants $c$ and $C$ independent of $\hbar$ such that
	\begin{equation*}
		c\Nrm{\op}{\dot{\cW}^{1,p}} \leq \Nrm{(-\DDh_x)^{1/2}\op}{\L^p} + \Nrm{(-\DDh_\xi)^{1/2}\op}{\L^p} \leq C \Nrm{\op}{\dot{\cW}^{1,p}}.
	\end{equation*}
	When $p\neq 2$, it is however not clear that these expressions are equivalent to $\Nrm{\op}{\dot{\cH}^{1,p}} = \Nrm{(-\DDh)^{1/2}\op}{\L^p}$ where $(-\DDh)^s$ was defined in Equation~\eqref{eq:lap_frac_def}. Hence, it is not clear that the quantum Bessel and Sobolev norms of order $1$ are equivalent. We conjecture however that this should be true, and more precisely that $(-\DDh)^{-1/2}\,\Dh : \L^p \to (\L^p)^{2d}$ is a bounded operator for any $p\in(1,\infty)$ with norm independent of $\hbar$.

\section{Semiclassical convolution}\label{sec:semiclassical_convolution}

	One of the important tools in our analysis will be the use of a quantum analogue of the convolution, actually already introduced in~\cite{werner_quantum_1984}, that we present in this section.

\subsection{Phase space translations}

	Let us start by recalling some well-known properties of translation operators in quantum mechanics. For any $z_0=(x_0,\xi_0)\in\Rdd$, we define the associated \textit{phase space translation operator} by $\tau_{z_0} = e^{i\(\xi_0\cdot x -x_0\cdot\opp\)/\hbar}$, i.e. $\tau_{z_0}\varphi = e^{i\,\xi_0\cdot\(x -\frac{x_0}{2}\)/\hbar}\, \varphi(x-x_0)$. The translation operators are unitary with $\tau_z^{-1} = \tau_{-z}$. They satisfy a semigroup formula up to a phase factor
	\begin{equation}\label{eq:translations_product}
		\tau_{z+z'} = e^{-i\pi\(x\cdot\xi'-x'\cdot\xi\)/h}\,\tau_{z}\tau_{z'}.
	\end{equation}
	On the operator level, the phase space translation of a density operator $\op$ is defined by Equation~\eqref{eq:translation_def}, which can be written $\sfT_z\op = \tau_{z} \, \op \, \tau_{-z}$. The semigroup relation~\eqref{eq:translations_product} now becomes a true semigroup formula
	\begin{equation*}
		\sfT_{z+z'}\op = \sfT_{z}\sfT_{z'}\op.
	\end{equation*}
	Notice that these operators translate the position and momentum operators through the formulas $\sfT_{z_0} \, x = x - x_0$ and  $\sfT_{z_0} \, \opp = \opp - \xi_0$, or equivalently, if $\opz = (x,\op)$, $\sfT_{z_0} \, \opz = \opz - z_0$. More generally, they translate any Weyl quantization, i.e.
	\begin{equation}\label{eq:Weyl_translation}
		\sfT_{z_0} \, \op_f = \op_{f(\cdot-z_0)}.
	\end{equation}
	
\subsection{Definition of the semiclassical convolution and inequalities}

	As was already introduced by Werner in~\cite{werner_quantum_1984}, we define the semiclassical convolution of a phase space function $f$ and a bounded operator $\op$ as the following operator-valued integral\footnote{See Remark~\ref{rmk:semiclassical_convolution_def} for the more rigorous meaning of this integral.}
	\begin{equation}\label{eq:semiclassical_convolution_def}
		f \star \op = \op \star f := \intdd f(z)\,\sfT_z \op\d z.
	\end{equation}
	One of the important property of this definition for the semiclassical convolution that was proved in~\cite{werner_quantum_1984} is the fact that the analogue of Young's convolution inequality holds. It also allows to better understand when the semiclassical convolution is well defined.
	
	\begin{thm}[\textbf{Semiclassical Young's convolution inequality} \cite{werner_quantum_1984}]\label{thm:semiclassical_convolution}
		Let $(p,q,r)\in[1,\infty]^3$ be such that $1 + \frac{1}{p} = \frac{1}{q} + \frac{1}{r}$. Then there exists $C>0$ independent of $\hbar$ such that for any $f\in L^q$
		\begin{equation}\label{eq:semiclassical_convolution}
			\Nrm{f\star\op}{\L^p} \leq \Nrm{f}{L^q} \Nrm{\op}{\L^r}.
		\end{equation}
		When $r=p$, one can also take $f$ to be a bounded measure, and replace the $L^q$ norm by the total variation norm.
	\end{thm}
	
	In our case, it will be useful to also have an analogue of the Hardy--Littlewood--Sobolev's inequality in the phase space, as given by the following theorem.
	
	\begin{thm}[\textbf{Semiclassical Hardy--Littlewood--Sobolev's inequality}]\label{thm:quantum_HLS}
		Let $(p,q,r)\in(1,\infty)^3$ be such that $1 + \frac{1}{p} = \frac{1}{q} + \frac{1}{r}$ and $r\leq p$. Then there exists $C>0$ independent of $\hbar$ such that for any $f\in L^{q,\infty}$
		\begin{equation}\label{eq:quantum_HLS}
			\Nrm{f\star\op}{\L^p} \leq C \Nrm{f}{L^{q,\infty}} \Nrm{\op}{\L^r}.
		\end{equation}
	\end{thm}
	
	\begin{remark}
		When moreover $(p,q,r)\in (1,\infty)^3$, then one more precisely has for any $(a,b,c)\in [1,\infty]^3$ satisfying $\frac{1}{c} = \frac{1}{a} + \frac{1}{b}$
		\begin{equation*}
			\Nrm{f\star\op}{\L^{p,c}} \leq C \Nrm{f}{L^{q,a}} \Nrm{\op}{\L^{r,b}}
		\end{equation*}
		where the spaces $L^{p,q} = L^{p,q}(\Rdd)$ are the Lorentz spaces, which are the spaces obtained by real interpolation of Lebesgue spaces, and the $\L^{p,q}$ spaces are their analogue for operators introduced in~\cite{birman_estimates_1975}.
		In particular, taking $a = \frac{p+1}{p}\,q$, $b = \frac{p+1}{p}\,r$ and using the inclusions between $\L^{p,a}$ spaces, one gets $\Nrm{f\star\op}{\L^{p,1}} \leq C \Nrm{f}{L^q} \Nrm{\op}{\L^{r}}$, which is slightly stronger than Equation~\eqref{eq:semiclassical_convolution}.
	\end{remark}
	
	\begin{remark}
		As in the classical case, by using the dual definition of the $\L^p$ norm, Inequality~\eqref{eq:quantum_HLS} can be written $h^d \Tr{\opmu\(f\star\op\)} \leq C \Nrm{f}{L^{r,\infty}} \Nrm{\opmu}{\L^p} \Nrm{\op}{\L^q} $ for any $(p,q,r)\in(1,\infty)^3$ such that $\frac{1}{p} + \frac{1}{q} + \frac{1}{r} = 2$. Taking $f(z) = \n{z}^{-a}$ leads by definition of the semiclassical convolution to
		\begin{equation*}
			h^d \tr\intdd \frac{\opmu\,\sfT_z\op}{\n{z}^a}\d z \leq C \Nrm{\opmu}{\L^p} \Nrm{\op}{\L^q},
		\end{equation*}
		whenever $\frac{1}{p} + \frac{1}{q} + \frac{a}{2d} = 2$.
		Because of the use of interpolation techniques, it is however not clear what the value of the constant $C$ is.
	\end{remark}
	
	From Formula~\eqref{eq:Weyl_translation} and the linearity of the Weyl quantization, we deduce that $f\star \op_g = \intdd f(z)\,\op_{g(\cdot-z)}\d z = \op_{f*g}$. In particular, the Wigner transform of the semiclassical convolution is given by the classical convolution through the formula
	\begin{equation*}
		f_{g\star\op} = g * f_{\op}
	\end{equation*}
	and the property of commutation of the classical convolution product yields
	\begin{equation}\label{eq:commutation_convolution}
		f_1\star \op_{f_2} = \op_{f_1*f_2} = \op_{f_1} \star f_2.
	\end{equation}
	The semiclassical convolution has properties close to the classical convolution. It preserves the positivity: if $\op\geq 0$ and $f\geq 0$, then $f\star \op\geq 0$. Taking $f=1$ yields $1\star \op = \op_{1*f_{\op}} = h^d\Tr{\op} \mathbf{1}$ where $\mathbf{1}$ is the identity on $L^2$. Taking $f=\delta_0$ yields also as expected
	\begin{equation}\label{eq:convolution_Dirac}
		\delta_0 \star \op = \sfT_0\op = \op.
	\end{equation}
	From the first identity in Equation~\eqref{eq:commutation_convolution} and the compatibility of the quantum gradients with the Weyl quantization~\eqref{eq:weyl_quantization_gradients}, we also deduce that whenever $\nabla f \in L^1 + L^\infty$ and $\nabla\op$ is trace class
	\begin{equation*}
		\Dh\!\(f\star\op\) = \nabla f\star\op = f\star\Dh\op.
	\end{equation*}

\subsection{Convolution by a Gaussian: Husimi transform and coherent states}

	We define the semiclassical Gaussian by
	\begin{align*}
		g_h(z) = \(2/h\)^d e^{-\n{z}^2/\hbar}.
	\end{align*}
	When $\hbar\to 0$, the convolution by this function is an approximation of the identity. For a function $f=f(z)$, we will denote by
	\begin{equation*}
		\tilde{f}(z) := g_h * f(z).
	\end{equation*}
	This converges to $f$ when $\hbar\to 0$. In the particular case when $f$ is a Wigner transform of an operator $\op$, we obtain the so called Husimi transform (see e.g. \cite[Eq.~(25)]{lions_sur_1993})
	\begin{align*}
		\tilde{f}_{\op} = g_h * f_{\op}
	\end{align*}
	which is a nonnegative function whenever $\op$ is a positive operator.
	
	To get the quantum analogue of the convolution by a Gaussian, we use the above defined semiclassical convolution. We define
	\begin{equation*}
		\tildop := g_h \star \op.
	\end{equation*}
	Similarly as above in the particular case when $\op$ is a Weyl quantization, we get
	\begin{equation*}
		\tildop_f := g_h \star \op_f = \op_{\widetilde{f}} = f \star \op_{g_h}.
	\end{equation*}
	
	We also deduce that the Husimi transform of $\op$ is nothing but the Wigner transform of $\tildop$, i.e. $\tilde{f}_{\op} = f_{\tildop}$. It is well-known (see e.g. \cite{lions_sur_1993}) and not difficult to prove that the Weyl quantization of $g_h$ is a projection of the form $\op_{g_h} = h^{-d} \ket{\psi_h}\bra{\psi_h}$ on the Gaussian coherent state
	\begin{equation*}
		\psi_h(x) = \(\tfrac{2}{h}\)^{d/4} e^{-\n{x}^2/(2\hbar)}.
	\end{equation*}
	Hence,
	\begin{equation*}
		\tildop_f = h^{-d} \intdd f(z) \ket{\tau_{z}\psi_h}\bra{\tau_{z}\psi_h} \d z
	\end{equation*}
	is nothing but what is sometimes called a superposition of coherent sates, a Töplitz operator or the (Anti-)Wick quantization of $f$. Let us mention that links between T\"oplitz operators and convolutions were studied in~\cite{fulsche_correspondence_2020}. As it is a special case of a semiclassical convolution and since $\op_{g_h} \geq 0$, we recover the fact that it maps positive functions to positive operators. Moreover, as a corollary of the semiclassical Young's inequality and the fact that $\Nrm{\op_{g_h}}{\L^1} = 1$, we get the well-known bound
	\begin{equation}\label{eq:Toplitz_bound}
		\Nrm{\tildop_f}{\L^p} \leq \Nrm{f}{L^p(\Rdd)}.
	\end{equation}
	\begin{remark}
		One can more generally prove (see e.g.~\cite[Eq.~(61),~(59)]{lions_sur_1993}) that for any convex function $\Phi$ such that $\Phi(0) = 0$, $h^d \Tr{\Phi\!\(\tildop_f\)} \leq \intdd \Phi(f)$. An analogous identity holds for the Husimi transform (see e.g.~\cite[Eq.~(64)]{lions_sur_1993}), that is $\intdd \Phi(\tilde{f}_{\op}) \leq h^d \Tr{\Phi\!\(\op\)}$. In particular,
		\begin{equation}\label{eq:Husimi_vs_op}
			\Nrm{\tilde{f}_{\op}}{L^p(\Rdd)} \leq \Nrm{\op}{\L^p}.
		\end{equation}
	\end{remark}

\section{Proofs}

	The general method for the proof consists in continuing our work of analogy with the classical case and trying to translate the classical proofs concerning Sobolev spaces to their noncommutative counterpart for Schatten norms of operators. This is the reason why the analogous objects have first to be chosen carefully. However, all proofs do not seem to translate easily to the operator setting.
	
	More precisely, the proof of the classical Bessel--Sobolev inequalities can be deduced from Hardy--Littlewood--Sobolev's inequalities. However, there does not seem to be a simple definition of operator-valued convolution of operators which satisfies an analogue of the Hardy--Littlewood--Sobolev's inequalities for Schatten norms. This is why we will be dealing with semiclassical convolution instead.
	
	For the proof of the Gagliardo--Sobolev inequalities~\eqref{thm:Gag_Sob}, and in particular for the proofs of the Sobolev inequalities with $s=p=1$, we were not able to obtain an analogue of the original proof of Gagliardo and Nirenberg \cite{gagliardo_ulteriori_1959, nirenberg_elliptic_1959} which uses the fundamental theorem of calculus on bounded domains, and we instead use the fact that Lebesgue norms of Husimi transforms are close to Schatten norms with errors controlled by quantum Besov norms.

\subsection{Semiclassical convolution}\label{sec:semiclassical_convolution_proofs}
	
	To prove the semiclassical analogue of the Hardy--Littlewood--Sobolev inequality, we will use the following bilinear interpolation result (see e.g.~\cite[Lemma~28.2]{tartar_introduction_2007} and \cite[10.1]{calderon_intermediate_1964}).
	\begin{lem}\label{lem:bilinear_interpolation}
		Let $B$ be a bilinear operator mapping continuously $E_0\times F_0 \to G_0$, $E_0\times F_1 \to G_1$, and $E_1\times F_0 \to G_1$. Then for every $\(\theta,\theta_2\)\in (0,1)^2$ such that $\theta + \theta_2 < 1$ and any $(a,b,c)\in [1,\infty]^3$ satisfying $\frac{1}{c} = \frac{1}{a} + \frac{1}{b}$, it holds
		\begin{align*}
			\Nrm{B(u,v)}{G_{\theta+\theta_2,c}} &\leq C \Nrm{u}{E_{\theta,a}} \Nrm{v}{F_{\theta_2,b}}
			\\
			\Nrm{B(u,v)}{G_\theta} &\leq M_0^{1-\theta} M_1^{\theta} \,\Nrm{u}{E_{\theta}} \Nrm{v}{F_\theta}
		\end{align*}
		where $X_{\theta,a} = (X_0,X_1)_{\theta,a}$, $X_{\theta} = [X_0,X_1]_{\theta}$ are respectively the real and complex interpolation spaces between $X_0$ and $X_1$, and $M_k$ denotes the norm of $B$ as a bilinear operator from $E_k\times F_k$ to $G_k$.
	\end{lem}
	
	The proof is then similar to the proof of the semiclassical Young's inequality in~\cite{werner_quantum_1984} and consists mainly in obtaining the estimate at the three endpoints.
	
	\begin{proof}[Proof of theorems~\ref{thm:semiclassical_convolution} and~\ref{thm:quantum_HLS}]
		Notice first that if $f\in L^1$ and $\op\in \L^p$ for any $p\in[1,\infty]$, then
		\begin{equation}\label{eq:young_1}
			\Nrm{f\star\op}{\L^p} \leq \intdd \n{f(z)} \Nrm{\sfT_z \op}{\L^p}\d z = \Nrm{f}{L^1} \Nrm{\op}{\L^p}.
		\end{equation}
		On another side, if $f$ is bounded and $\op\in\L^1$, then for any $(\varphi,\phi)\in (L^2(\Rd))^2$, by definition it holds
		\begin{equation}\label{eq:semiclassical_convolution_action}
			\Inprod{\varphi}{\(f \star \op\)\phi} = \intdd f(z) \Inprod{\varphi}{\(\sfT_z \op\)\phi}\d z.
		\end{equation}
		Since $\op$ is a compact operator, we may write it using its singular values decomposition (see e.g. \cite{reed_functional_1980, simon_trace_2005}) under the form
		\begin{equation*}
			\op = \sumj \mu_j \ket{\psi_j}\bra{\tilde{\psi}_j}
		\end{equation*}
		for some orthonormal families $(\psi_j)_{j\in J}$ and $(\tilde{\psi}_j)_{j\in J}$ of functions in $L^2(\Rd)$, and where $(\mu_j)_{j\in J}$ are positive numbers called the singular values of $\op$. Since $\op$ is trace class, the singular values form an integrable sequence. Thus, applying the Fubini theorem and then the Cauchy--Schwarz inequality to Equation~\eqref{eq:semiclassical_convolution_action} leads to
		\begin{align*}
			\Inprod{\varphi}{\(f \star \op\)\phi} &\leq \sumj \mu_j \Nrm{f}{L^\infty} \intdd \n{\Inprod{\varphi}{\tau_z\psi_j} \Inprod{\tau_z\tilde{\psi}_j}{\phi}}\d z
			\\
			&\leq \Tr{\n{\op}} \Nrm{f}{L^\infty} \sup_j \Nrm{\Inprod{\varphi}{\tau_z\psi_j}}{L^2_z} \Nrm{\Inprod{\tau_z\tilde{\psi}_j}{\phi}}{L^2_z}.
		\end{align*}
		Observe that by the definition of the shift operators $\tau_z$, the inner products appearing on the right-hand side of the above inequality can be written as a Fourier transform under the form
		\begin{equation*}
			\Inprod{\varphi}{\tau_z\psi_j}
			= e^{-i\xi\cdot x/(2\hbar)} \F{\varphi(\cdot)\,\psi_j(\cdot-x)}(\tfrac{\xi}{h}).
		\end{equation*}
		Therefore, it follows from the Plancherel theorem and the normalization of $\psi_j$ that
		\begin{equation*}
			\Nrm{\Inprod{\varphi}{\tau_z\psi_j}}{L^2_z} = h^{d/2}\Nrm{\varphi(\xi)\,\psi_j(\xi-x)}{L^2_{x,\xi}} = h^{d/2}\Nrm{\varphi}{L^2}.
		\end{equation*}
		The same inequality holds when replacing $\psi_j$ by $\tilde{\psi}_j$, and $\varphi$ by $\phi$. This leads to $\Inprod{\varphi}{\(f \star \op\)\phi} \leq h^d\Tr{\n{\op}} \Nrm{f}{L^\infty} \Nrm{\varphi}{L^2} \Nrm{\phi}{L^2}$. Hence, $f \star \op$ is a bounded operator with norm satisfying
		\begin{align*}
			\Nrm{f \star \op}{\L^\infty} \leq \Nrm{f}{L^\infty} \Nrm{\op}{\L^1}.
		\end{align*}
		Theorem~\ref{thm:semiclassical_convolution} now follows by combining the above equation with Inequality~\eqref{eq:young_1} and bilinear complex interpolation as in Lemma~\ref{lem:bilinear_interpolation}. Theorem~\ref{thm:quantum_HLS} follows from the bilinear real interpolation result of the same lemma with $1-\theta=1/q$ and $1-\theta_2=1/r$.
	\end{proof}
	
	\begin{remark}\label{rmk:semiclassical_convolution_def}
		Let us give a comment about the signification of the definition of the semiclassical convolution. In general, if $f\in L^1(\Rdd)+L^\infty(\Rdd)$ and $\op$ is a bounded operator, then the integral appearing of Formula~\eqref{eq:semiclassical_convolution_def} as to be considered as a weak (Pettis) integral. Interpreting a bounded operator $\op$ acting on $L^2$ as a linear form on $L^2(\Rd)\otimes L^2(\Rd)$ through the formula $\scalar{\op}{\varphi\otimes\phi} := \Inprod{\conj{\varphi}}{\op\phi}$, then a correct interpretation of Definition~\eqref{eq:semiclassical_convolution_def} is
		\begin{equation*}
			\scalar{f \star \op}{\varphi\otimes\phi} = \intdd \scalar{f(z)\,\sfT_z \op}{\varphi\otimes\phi}\d z.
		\end{equation*}
		This is indeed a well defined Lebesgue integral as follows from the above proof whenever $f\in L^\infty$. In the case when $f\in L^1$, then by Inequality~\eqref{eq:young_1}, the integral in Definition~\eqref{eq:semiclassical_convolution_def} can even better be defined as a strong (Bochner) integral, and so as a weak integral.
	\end{remark}

\subsection{Proofs of Sobolev inequalities}

	We can now prove our main theorems. We start with the proof of the inequality for Bessel type spaces without giving an explicit constant. It is the simplest proof as it is a consequence of the Hardy--Littlewood--Sobolev inequality~\eqref{eq:quantum_HLS}. The method of proof is actually reminiscent of the original strategy of Sobolev~\cite{sobolev_theorem_1938}. It however does not include the case $p=1$.
	
	\begin{prop}
		Let $s > 0$ and $1< p\leq q < \infty$ satisfy \eqref{eq:exponent_sobolev}. Then there exist a constant $C$ independent of $\hbar$ such that for any compact operator $\op$,
		\begin{align}\label{eq:Sobolev-Bessel_1}
			\Nrm{\op}{\L^q} &\leq C \Nrm{(-\DDh)^{s/2}\op}{\L^p} && \text{ if } s\in (0,2d)
			\\\label{eq:Sobolev-Bessel_2}
			\Nrm{\op}{\L^q} &\leq C \Nrm{(-\DDh)^\frac{s-1}{2}\Dh\op}{\L^p} && \text{ if } s\in (0,2d-1).
		\end{align}
	\end{prop}
	
	\begin{remark}
		As a particular case of Inequality~\eqref{eq:Sobolev-Bessel_2}, when $s=1$ and $1<p\leq q<\infty$ satisfy $\frac{1}{2d} = \frac{1}{p} - \frac{1}{q}$, then we already obtain the Sobolev inequality~\eqref{eq:Sobolev}.
	\end{remark}

	\begin{proof}
		For $s \in\R_+\setminus\{2d\}$, define the function acting on $z\in\Rdd$ by
		\begin{equation*}
			K_s(z) = \frac{c_{2d,-s}}{\n{z}^{2d-s}}
		\end{equation*}
		with $c_{2d,-s} = \tfrac{\omega_s}{\(2\pi\)^s\omega_{2d-s}}$. In the case when $s=2d$, define instead $K_{2d}(z) = -\omega_{2d}\ln(\n{\pi x})$. When $s\in(0,2d)$, it is well-known that this function is a solution in the sense of distributions of
		\begin{equation*}
			(-\Delta_z)^{s/2}K_s = \delta_0
		\end{equation*}
		as can be checked for example taking the Fourier transform. Since $K_s \in L^{\frac{2d}{2d-s},\infty}$ when $s\in(0,2d)$ and $1+\frac{1}{q} = \frac{2d-s}{2d} + \frac{1}{p}$ with $p\leq q$, it follows from Formula~\eqref{eq:convolution_Dirac} and the quantum Hardy--Littlewood--Sobolev inequality~\eqref{eq:quantum_HLS} that
		\begin{align*}
			\Nrm{\op}{\L^q} &= \Nrm{(-\Delta_z)^{s/2}K_{s}\star\op}{\L^q} = \Nrm{K_{s}\star (-\DDh)^{s/2}\op}{\L^q} \leq C \Nrm{(-\DDh)^{s/2}\op}{\L^p}
		\end{align*}
		which proves Inequality~\eqref{eq:Sobolev-Bessel_1}. In the same spirit, when $s\in(0,2d-1)$, the vector valued function $\nabla K_{s+1} \in L^{\frac{2d}{2d-s},\infty}$ satisfies $\(-\Delta_z\)^\frac{s-1}{2}\nabla_z\cdot\(\nabla_z K_{s+1}\) = -\delta_0$, and one obtains Inequality~\eqref{eq:Sobolev-Bessel_2}.
	\end{proof}

	We now prove the Sobolev inequalities for the fractional Sobolev spaces $\dot{\cW}^{s,p}$ as defined in Equation~\eqref{eq:Sobolev_frac_def}. This proves Theorem~\ref{thm:Gag_Sob}, except for the special bounds on the constant $\CSh_{s,p}$ in the case when $p=2$ or $q=2$, which will be proved in lemmas~\ref{lem:Sobolev_01} and~\ref{lem:Sobolev_02}, and the lower bound that will be proved in Lemma~\ref{lem:lower_bound}.

	\begin{prop}
		Let $s\in [0,1]$ and $1\leq p\leq q < \infty$ satisfy~\eqref{eq:exponent_sobolev}. Then
		\begin{equation*}
			\Nrm{\op}{\L^q} \leq \CSh_{s,p} \Nrm{\op}{\dot{\cW}^{s,p}}
		\end{equation*}
		where $\CSh_{s,p}$ satisfies Inequality~\eqref{eq:Gag_Sob_cst_frac} when $s\in (0,1)$ and Inequality~\eqref{eq:Gag_Sob_cst_1} when $s=1$.
	\end{prop}

	\begin{proof}
		It follows from the triangle inequality for Schatten norms that
		\begin{equation}\label{eq:double_convolution_expansion}
			\Nrm{\op}{\L^q} \leq \Nrm{\ttildop}{\L^q} +  \Nrm{\ttildop-\op}{\L^q}
		\end{equation}
		where $\ttildop = g_h\star(g_h\star\op) = (g_h*g_h)\star\op = g_{2h}\star\op$. To bound the first term, notice that $\ttildop = \tildop_{\tilde{f}_{\op}}$ where $\tilde{f}_{\op}$ is a smooth function converging to $0$ at infinity. The smoothness comes from the convolution by the Gaussian. The decay at infinity follows from the fact that compact operators are limit in $\L^\infty$ of finite rank operators, and so their Husimi transform is the uniform limit of integrable functions. Therefore, by the classical Sobolev inequality and by the property~\eqref{eq:Toplitz_bound} of the Wick quantization
		\begin{equation*}
			\Nrm{\ttildop}{\L^q} \leq \Nrm{\tilde{f}_{\op}}{L^q} \leq \CS_{s,p} \Nrm{\tilde{f}_{\op}}{\dot{W}^{s,p}}.
		\end{equation*}
		Then, if $s=1$, by the property~\eqref{eq:Husimi_vs_op} of the Husimi transform,
		\begin{equation}\label{eq:convolution_and_Sobolev_1}
			\Nrm{\tilde{f}_{\op}}{\dot{W}^{1,p}} = \Nrm{\tilde{f}_{\Dh\op}}{L^p} \leq \Nrm{\Dh\op}{\L^p}.
		\end{equation}
		Similarly, if $s\in (0,1)$, by the definition of the fractional Sobolev spaces and property~\eqref{eq:Husimi_vs_op},
		\begin{equation}\label{eq:convolution_and_Sobolev_frac}
			\Nrm{\tilde{f}_{\op}}{\dot{W}^{s,p}}^p = \gamma_{s,p} \Nrm{\frac{\Nrm{\tilde{f}_{\sfT_{z}\op - \op}}{L^p}}{\n{z}^{s+2d/p}}}{L^p}^p \leq \gamma_{s,p} \Nrm{\frac{\Nrm{\sfT_{z}\op - \op}{\L^p}}{\n{z}^{s+2d/p}}}{L^p}^p = \Nrm{\op}{\dot{\cW}^{s,p}}^p.
		\end{equation}
		We now obtain a bound on the second term in the right-hand side of Inequality~\eqref{eq:double_convolution_expansion}. Since $\intdd g_{2h} = 1$ and $g_{2h}$ is even, it holds
		\begin{equation*}
			\ttildop - \op = \frac{1}{2}\intdd g_{2h}(z) \(\sfT_z\op + \sfT_{-z}\op - 2\,\op\) \d z.
		\end{equation*}
		Therefore, taking the $\L^q$ norm leads to
		\begin{equation*}
			\Nrm{\ttildop - \op}{\L^q} \leq \frac{1}{2}\intdd g_{2h}(z) \n{z}^{s+2d/r} \frac{\Nrm{\sfT_z\op + \sfT_{-z}\op - 2\,\op}{\L^q}}{\n{z}^{s+2d/r}} \d z.
		\end{equation*}
		By H\"older's inequality, this leads to
		\begin{equation}\label{eq:convolution_error_quantum}
			\Nrm{\ttildop - \op}{\L^q} \leq \tfrac{1}{2}\,C_{d,s,r}\,h^{s/2} \Nrm{\op}{\dot{\cB}^s_{q,r}}
		\end{equation}
		where
		\begin{equation*}
			C_{d,s,r} = \frac{1}{h^{s/2}} \Nrm{g_{2h}(z) \n{z}^{s+\frac{2d}{r}}}{L^{r'}_z} = \frac{1}{\(r'\)^{d+\frac{s}{2}}} \(\frac{\omega_{2d}}{\omega_{(2d+s)r'}}\)^\frac{1}{r'}.
		\end{equation*}
		Observe now that for any $p\leq q$, by the scaled definition of the Schatten norms~\eqref{eq:def_norm}, it holds $h^{s/2}\Nrm{\op}{\L^q} = h^{d\(\frac{1}{q}-\frac{1}{p}\)+\frac{s}{2}} \Nrm{\op}{\L^p}$. Therefore, since $\frac{1}{p}-\frac{1}{q} = \frac{s}{2d}$, it leads to
		\begin{equation*}
			\Nrm{\ttildop - \op}{\L^q} \leq \tfrac{1}{2}\,C_{d,s,r} \Nrm{\op}{\dot{\cB}^s_{p,r}}.
		\end{equation*}
		When $s=1$, then we use Inequality~\eqref{eq:comparison_besov_sobolev} to deduce that
		\begin{equation}\label{eq:convolution_error_quantum_12}
			\Nrm{\ttildop - \op}{\L^q} \leq C_{d,s,\infty} \Nrm{\op}{\dot{\cW}^{s,p}}
		\end{equation}
		while if $s\in(0,1)$, then we use Proposition~\ref{prop:comparison_besov_sobolev_frac} to get
		\begin{equation*}
			\Nrm{\ttildop - \op}{\L^q}^p \leq \tfrac{\(1-2^{\(s-1\)r}\)^{p/r}}{\gamma_{s,p}}\,C_{d,s,p}^p \Nrm{\op}{\dot{\cW}^{s,p}}^p
		\end{equation*}
		where $r = \max(p,p')$. In any of these cases, combining the estimates on $\ttildop$ and $\ttildop  - \op$ finishes the proof.
	\end{proof}
	
	We now get explicit upper bounds for the Bessel--Sobolev inequality in the case when $p \leq 2 \leq q$, i.e. when $q\in [2,\frac{2\,d}{d-1}]$, using the fact that the Wigner transform is an isometry from $\L^2$ to $L^2(\Rdd)$.
	\begin{prop}\label{prop:Bessel-sobolev_cst}
		With the notations of theorems~\ref{thm:Gag_Sob} and \ref{thm:Bes_Sob}, assume $1 < p \leq 2 \leq q < \infty$ satisfy Equation~\eqref{eq:exponent_sobolev}, and let $\tilde{s} := 2d\,(\frac{1}{2} - \frac{1}{q})$. Then
		\begin{equation*}
			\CBh_{s,p} \leq \CBh_{\tilde{s},2} \, \CBh_{s-\tilde{s},p} \leq \(\CB_{\tilde{s},2} + \tfrac{\theta_{\tilde{s}}}{\(8\pi\)^{\tilde{s}/2}}\) \Big(\CB_{s-\tilde{s},p} + \tfrac{\theta_{s-\tilde{s}}}{\(8\pi\)^{(s-\tilde{s})/2}}\Big)
		\end{equation*}
		and similarly, if $s=1$, then the same inequality holds with $\CBh$ replaced by $\CSh$ when $1 < p \leq 2 \leq q < \infty$.
	\end{prop}
	
	\begin{remark}
		Recall that $\theta_s = \sup_{r\geq 0} \frac{1-e^{-r}}{r^s}$. In the case when $s<1$, then $\theta_s < 1$ and it is reached for some $r>0$ satisfying $e^r = 1+r/s$.
	\end{remark}
	
	We start with the special cases $p=2$ and $q=2$. The proposition then follows by taking the case $\tilde{p} = 2$ as an intermediate case, and writing
	\begin{equation*}
		\Nrm{\op}{\L^q} \leq \CBh_{\tilde{s},2} \Nrm{\op}{\dot{\cH}^s} = \CBh_{\tilde{s},2} \Nrm{(-\DDh)^{\tilde{s}/2}\op}{\L^2} \leq \CBh_{\tilde{s},2}\,\CBh_{s-\tilde{s},p} \Nrm{(-\DDh)^{s/2}\op}{\dot{\L}^p}.
	\end{equation*}
	
	\begin{lem}[First basis case]\label{lem:Sobolev_01}
		Let $s\in[0,1]$ and $q := \frac{2\,d}{d-s}$. Then for any $f\in\dot{H}^s(\Rdd)$,
		\begin{equation}\label{eq:Sobolev_1}
			\Nrm{\op_f}{\L^q} \leq \CSh_{s,2} \Nrm{f}{\dot{H}^s(\Rdd)} = \CSh_{s,2} \Nrm{\op_f}{\dot{\cH}^{s}}
		\end{equation}
		where $\CSh_{s,2} = \CBh_{s,2} \leq \CB_{s,2} + \theta_s \(8\pi\)^{-\frac{s}{2}}$ with $\CS_{s,2} = \CB_{s,2}$ the optimal constant appearing in the associated classical Sobolev inequality \eqref{eq:optimal_bessel_cst}.
	\end{lem}
	
	\begin{proof}
		It follows from the definition of $\L^2$ given in Equation~\eqref{eq:def_norm}, the Plancherel theorem and the fact that the Fourier transform of $g_h$ is $\widehat{g_h}(w) = e^{-\pi h\n{w}^2/2}$ that
		\begin{equation*}
			\Nrm{\op_f - \tildop_f}{\L^2} = \Nrm{f - g_h * f}{L^2(\Rdd)} = \Nrm{u(\tfrac{\pi h \n{w}^2}{2}) \(\tfrac{\pi h \n{w}^2}{2}\)^\frac{s}{2}\widehat{f}}{L^2(\Rdd)},
		\end{equation*}
		where $u : \R_+\to\R_+$ be defined by $u(r) = \frac{1-e^{-r}}{r^{s/2}}$. Since $u$ is a bounded function on $\R_+$, it yields
		\begin{equation}\label{eq:convolution_error}
			\Nrm{\op_f - \tildop_f}{\L^2} \leq \theta_s \(\tfrac{h}{8\pi}\)^{s/2} \Nrm{f}{\dot{H}^s(\Rdd)}.
		\end{equation}
		From the relation $\frac{1}{q} = \frac{1}{2} - \frac{s}{2d}$, the scaling of the $\L^p$ norms and the inequalities between Schatten norms, the above inequality leads to
		\begin{align*}
			\Nrm{\op_f - \tildop_f}{\L^q} = h^\frac{d}{q} \Nrm{\op_f - \tildop_f}{q} \leq h^{-\frac{s}{2}} \Nrm{\op_f - \tildop_f}{\L^2} \leq \tfrac{\theta_s}{\(8\pi\)^{s/2}} \Nrm{f}{\dot{H}^s(\Rdd)}.
		\end{align*}
		By the Triangle inequality for the Schatten norms and the bound~\eqref{eq:Toplitz_bound} for the Wick quantization, it leads to the following inequality
		\begin{equation*}
			\Nrm{\op_f}{\L^q} \leq \Nrm{\tildop_f}{\L^q} + \tfrac{\theta_s}{\(8\pi\)^{s/2}} \Nrm{f}{\dot{H}^s(\Rdd)} \leq \Nrm{f}{L^q(\Rdd)} + \tfrac{\theta_s}{\(8\pi\)^{s/2}} \Nrm{f}{\dot{H}^s(\Rdd)}
		\end{equation*}
		and we deduce Inequality~\eqref{eq:Sobolev_1} by the classical Sobolev inequality.
	\end{proof}
	
	\begin{lem}[Second basis case]\label{lem:Sobolev_02}
		Let $s\in[0,1]$ and $p> 1$ be such that $\frac{1}{p} = \frac{1}{2}+\frac{1}{2d}$. Then $\CBh_{s,p} \leq \CB_{s,p} + \tfrac{\theta_s}{\(8\pi\)^{s/2}}$, and if $s=1$ and $p\geq 1$, then $\CSh_{s,p} \leq \CS_{s,p} + \tfrac{\theta_s}{\(8\pi\)^{s/2}}$.
	\end{lem}
	
	\begin{proof}
		Once again we use Inequality~\eqref{eq:convolution_error}, which can be written in terms of the Wigner transform $f_{\op}$ of $\op$ in the form
		\begin{equation*}
			\Nrm{f_{\op} - \tilde{f}_{\op}}{L^2(\Rdd)} \leq \theta_s \(\tfrac{h}{8\pi}\)^{s/2} \Nrm{(-\Delta)^\frac{s}{2} f_{\op}}{L^2(\Rdd)} = \theta_s \(\tfrac{h}{8\pi}\)^{s/2} \Nrm{(-\DDh)^\frac{s}{2} \op}{\L^2},
		\end{equation*}
		and so by the fact that $\Nrm{\opmu}{2} \leq \Nrm{\opmu}{p}$ since $p\leq 2$, the definition of $p$, and the definition of $\L^p$ norms, we deduce
		\begin{equation*}
			\Nrm{f_{\op} - \tilde{f}_{\op}}{L^2(\Rdd)} \leq \tfrac{\theta_s}{\(8\pi\)^{s/2}}
			 \Nrm{\op}{\dot{\cH}^{s,p}}.
		\end{equation*}
		Then by the classical Bessel--Sobolev inequalities~\eqref{eq:optimal_bessel_cst} and Inequality~\eqref{eq:Husimi_vs_op}, we get
		\begin{equation*}
			\Nrm{\tilde{f}_{\op}}{L^2(\Rdd)} \leq \CB_{s,p} \Nrm{(-\Delta)^\frac{s}{2}\tilde{f}_{\op}}{L^p(\Rdd)} = \CB_{s,p} \Nrm{\tilde{f}_{(-\DDh)^\frac{s}{2}\op}}{L^p(\Rdd)} \leq \CB_{s,p} \Nrm{(-\DDh)^\frac{s}{2}\op}{\L^p}.
		\end{equation*}
		As in the proof of Lemma~\ref{lem:Sobolev_01}, the triangle inequality and the two previous inequalities lead to a bound on $\Nrm{\op}{\L^2} = \Nrm{f_{\op}}{L^2(\Rdd)}$. When $s=1$, the bound on $\CSh$ follows similarly replacing $(-\DDh)^\frac{s}{2}$ by the quantum gradient $\Dh^s$.
	\end{proof}
	
\subsubsection{Bound by below for the quantum Sobolev constants}

	In this section, we finish the proof of Theorem~\ref{thm:Gag_Sob} by proving a lower bound for the constant $\CSh_{s,p}$.

	\begin{lem}\label{lem:lower_bound}
		Under the hypotheses of Theorem~\ref{thm:Gag_Sob}, assuming $p=2$ or $s=1$, then the optimal uniform-in-$\hbar$ constant appearing in inequality~\eqref{eq:Sobolev-Gagliardo} satisfies
		\begin{equation*}
			\CSh_{s,p} \geq \CS_{s,p}.
		\end{equation*}
	\end{lem}
	
	\begin{proof}
		When $p=2$ or $s=1$, it is well-known \cite{talenti_best_1976, lieb_sharp_1983} that the optimizer of the Sobolev inequality~\eqref{eq:optimal_sobolev_cst} is given by the function
		\begin{equation*}
			g_{s,p}(z) = \frac{1}{\(1+\n{z}^{p'}\)^{2d/q}}
		\end{equation*}
		up to rotation, translation and dilatation. This means that $g\in \dot{W}^{s,p}(\Rdd)$ and $\Nrm{g_{s,p}}{L^q} = \CS_{s,p}\Nrm{g_{s,p}}{\dot{W}^{s,p}}$. The observation that we will use is that these functions are actually in a slightly more regular space than $\dot{W}^{s,p}$. Letting $a=p'$ and $b = 2d/q$, then by an explicit computation, it is not difficult to see that for any $n\in\N$ and any $x\neq 0$,
		\begin{equation*}
			\n{\nabla^n g} \leq C\, \frac{ \n{z}^{a-n} + \n{z}^{n\(a-1\)}}{\(1+\n{z}^a\)^{b+n}}
		\end{equation*}
		for some constant $C > 0$ depending only one $d$, $a$, $b$ and $n$. Therefore, $g\in \dot{W}^{n,r}(\Rdd)$ as soon as $a > n-2d/r$, and $ab > 2d/r - n$. We can then estimate by above the $\dot{W}^{\alpha,q}(\Rdd)$ norm of $g_{s,p}$ by the classical Sobolev embedding $\dot{W}^{n,r}(\Rdd) \subset \dot{W}^{\alpha,q}(\Rdd)$ with $n = \lceil \alpha\rceil$. This leads to
		\begin{equation*}
			g_{s,p}\in \dot{W}^{\alpha,q}(\Rdd) \text{ for any } 0 \leq \alpha  < p' + \frac{2d}{q}.
		\end{equation*}
		Now consider $\op := \tildop_{g_{s,p}}$. Then by the properties of the Husimi transform,
		\begin{equation}\label{eq:lower_bound_1}
			\Nrm{\op}{\L^q} \geq \Nrm{\tilde{f}_{\op}}{L^q} = \Nrm{\tilde{\tilde{g}}_{s,p}}{L^q} \geq \Nrm{g_{s,p}}{L^q} - C \,h^{\alpha/2} \Nrm{g_{s,p}}{\dot{W}^{\alpha,q}}
		\end{equation}
		where the last inequality follows by the same proof as the proof of Inequality~\eqref{eq:convolution_error_quantum} together with the triangle inequality. The first term on the right-hand side of Equation~\eqref{eq:lower_bound_1} is bounded below using the fact that similarly as in~\eqref{eq:convolution_and_Sobolev_1} and~\eqref{eq:convolution_and_Sobolev_frac}
		\begin{equation}\label{eq:lower_bound_2}
			\Nrm{g_{s,p}}{L^q} = \CS_{s,p}\Nrm{g_{s,p}}{\dot{W}^{s,p}} \geq \CS_{s,p}\bNrm{\tildop_{g_{s,p}}}{\dot{\cW}^{s,p}} = \CS_{s,p}\Nrm{\op}{\dot{\cW}^{s,p}}.
		\end{equation}
		To control the second term on the right-hand side of Equation~\eqref{eq:lower_bound_1} by something of the form $-C h^{\alpha/2} \Nrm{\op}{\dot{W}^{s,p}}$, one similarly writes
		\begin{equation}\label{eq:lower_bound_3}
			\Nrm{\op}{\dot{W}^{s,p}} = \bNrm{\tildop_{g_{s,p}}}{\dot{W}^{s,p}} \geq \Nrm{\tilde{\tilde{g}}_{s,p}}{\dot{W}^{s,p}} \geq C
		\end{equation}
		where $C$ does not depend on $\hbar$. Combining the above equations~\eqref{eq:lower_bound_1},~\eqref{eq:lower_bound_2} and~\eqref{eq:lower_bound_3} yields
		\begin{equation*}
			\(\CS_{s,p} - C\,h^{\alpha/2}\)\Nrm{\op}{\dot{W}^{s,p}} \leq \Nrm{\op}{\L^q}
		\end{equation*}
		for some constant $C$ depending only on $(d,s,p,\alpha)$. Hence, the constant $\CSh_{s,p}$ satisfies $\CSh_{s,p} \geq \CS_{s,p} - C\,h^{\alpha/2}$ for any $h>0$. Since $\CSh_{s,p}$ is by definition independent of $h$, we deduce the result by letting $h\to 0$.
	\end{proof}
	
\subsection{Quantum Besov spaces}

	Define  $\delta^1_z\op := \sfT_z\op - \op$ and $\delta^2_z\op := \delta^1_z\delta^1_z\op = \sfT_{2z}\op - 2 \sfT_z\op + \op$. When $s\in(0,1)$, the following norm
	\begin{equation}\label{eq:Bspq_1}
		\Nrm{\op}{\dot{\cB}^s_{p,r,(1)}} :=  \BNrm{\frac{\sfT_z\op - \op}{\n{z}^{s+2d/r}}}{L^r_z(\L^p)}
	\end{equation}
	is equivalent to the $\dot{\cB}^s_{p,r}$ norm, independently of $\hbar$. Proposition~\ref{prop:comparison_besov_sobolev_frac} corresponds to the special case $p=q$ in the following proposition comparing these equivalent norms.
	
	\begin{prop}\label{prop:Besov_differences_2}
		Let $(s,p,q)\in (0,1)\times [1,\infty]^2$ and $\op\in \dot{\cB}^s_{p,q}$. Then $\op\in\dot{\cB}^s_{p,q,(1)}$ and
		\begin{equation}\label{eq:Besov_differences_2}
			(2^{r'} - 2^{sr'})^\frac{1}{r'} \Nrm{\op}{\dot{\cB}^s_{p,q,(1)}} \leq \Nrm{\op}{\dot{\cB}^s_{p,q}} \leq \(2^r - 2^{sr}\)^\frac{1}{r} \Nrm{\op}{\dot{\cB}^s_{p,q,(1)}}
		\end{equation}
		where $r = \max(p,p',q,q')$.
	\end{prop}
	
	\begin{remark}
		We could also have defined $r$ by $\frac{1}{2} - \frac{1}{r} = \max\!\(\n{\frac{1}{2}-\frac{1}{p}}, \n{\frac{1}{2}-\frac{1}{q}}\)$. Similarly, $r'$ satisfies $r' = \min(p,p',q,q')$. In all the cases, $r' \leq 2 \leq r$, with equality only if $r=2$.
	\end{remark}
	
	\begin{proof}[Proof of Proposition~\ref{prop:Besov_differences_2} and Proposition~\ref{prop:comparison_besov_sobolev_frac}]
		Notice that
		\begin{align*}
			\delta_z^2 \op &= \delta_z^1 \sfT_z\op - \delta_z^1 \op
			\\
			\delta_{2z}^1 \op &= \delta_z^1 \sfT_z\op + \delta_z^1 \op.
		\end{align*}
		Let $\opmu_1 = \n{\delta_z^2 \op}$, $\opmu_2 = \n{\delta_{2z}^1 \op}$, $\opnu_1 = \n{\delta_z^1 \sfT_z\op}$ and $\opnu_2 = \n{\delta_z^1 \op}$. Then, defining $\mu(\d z) := \frac{\d z}{\n{z}^{d+sq}}$, and the notation $\Nrm{\opmu}{L^q_\mu \L^p_{\vphantom\mu}}^q := \intd \Nrm{\opmu(z)}{\L^p}^q\mu(\d z)$, it holds
		\begin{align*}
			\Nrm{\opmu_1}{L^q_\mu \L^p_{\vphantom\mu}} &= \Nrm{\op}{\dot{\cB}^s_{p,q}}
			\\
			\Nrm{\opmu_2}{L^q_\mu \L^p_{\vphantom\mu}} &= 2^s \Nrm{\op}{\dot{\cB}^s_{p,q,(1)}}
			\\
			\Nrm{\opnu_1}{L^q_\mu \L^p_{\vphantom\mu}} &= \Nrm{\opnu_2}{L^q_\mu \L^p_{\vphantom\mu}} =  \Nrm{\op}{\dot{\cB}^s_{p,q,(1)}}.
		\end{align*}
		Since $r\geq \max(p,p')$, with the notation $\opmu = (\opmu_1,\opmu_2)$ and $\opnu = (\opnu_1, \opnu_2)$ we deduce from the Clarkson--McCarthy inequalities (See Equation~\eqref{eq:Clarkson-McCarthy} in Appendix~\ref{appendix:vec_operators}) that
		\begin{equation*}
			\Nrm{\opmu}{\ell^r(\L^p)} \leq 2^{1/r} \Nrm{\opnu}{\ell^{r'}(\L^p)}.
		\end{equation*} 
		By Jensen's inequality, since $q\leq r$ and $r' \leq q$, this leads to
		\begin{equation*}
			\Nrm{\opmu}{\ell^{^r}\! L^q_\mu \L^{^p}} \leq \Nrm{\opmu}{L^q_\mu \ell^{^r}\! \L^{^p}} \leq 2^{1/r} \Nrm{\opnu}{L^q_\mu \ell^{^{r'}}\!\! \L^{^p}} \leq 2^{1/r} \Nrm{\opnu}{\ell^{^{r'}}\!\! L^q_\mu \L^{^p}}.
		\end{equation*}
		In terms of $\op$, this yields
		\begin{equation*}
			\(\Nrm{\op}{\dot{\cB}^s_{p,q}}^r + 2^{rs} \Nrm{\op}{\dot{\cB}^s_{p,q,(1)}}^r\)^{1/r}  \leq 2^{1/r}\(2\Nrm{\op}{\dot{\cB}^s_{p,q,(1)}}^{r'}\)^{1/r'} = 2 \Nrm{\op}{\dot{\cB}^s_{p,q,(1)}}
		\end{equation*}
		which proves the second inequality in Equation~\eqref{eq:Besov_differences_2}. The first inequality follows similarly using Jensen's inequality and the reversed version of the Clarkson--McCarthy inequalities (Formula~\eqref{eq:Clarkson-McCarthy_2} in Appendix~\ref{appendix:vec_operators}).
	\end{proof}
	
	\begin{proof}[Proof of Proposition~\ref{prop:comparison_besov_sobolev}]
		From the Formula~\eqref{eq:fundamental_theorem} and the fact that the unitary operators $\sfT_{\theta z}$ preserve Schatten norms, one obtains
		\begin{equation*}
			\Nrm{\delta^2_z\op}{\L^p} \leq 2\n{z}\Nrm{\op}{\dot{\cW}^{1,p}}
		\end{equation*}
		while using the triangle inequality yields $\Nrm{\delta^2_z\op}{\L^p} \leq 4\Nrm{\op}{\L^p}$. Therefore, cutting the integral appearing in the definition of the Besov norm into two parts $\n{z}< R$ and $\n{z}\geq R$ leads to
		\begin{equation*}
			\intdd \frac{\Nrm{\delta^2_z\op}{\L^p}^r}{\n{z}^{2d+s r}}\d z 
			\leq 2^r\omega_{2d}\(\frac{R^{\(1-s\)r}}{\(1-s\)r}\Nrm{\op}{\dot{\cW}^{1,p}}^r + \frac{2^r}{s\,r\,R^{s r}} \Nrm{\op}{\L^p}^r\)
		\end{equation*}
		optimizing with respect to $R$ leads to
		\begin{align*}
			\intdd \frac{\Nrm{\delta^2_z\op}{\L^p}^r}{\n{z}^{2d+sr}}\d z \leq \frac{2^{\(2-s\) r}\,\omega_{2d}}{r\,s\(1-s\)} \Nrm{\op}{\dot{\cW}^{1,p}}^{s r} \Nrm{\op}{\L^p}^{\(1-s\) r}
		\end{align*}
		which proves the result.
	\end{proof}
	
\subsubsection{Semiclassical Littlewood--Paley decomposition}\label{sec:Littlewood-Paley}

	We refer for example to~\cite{triebel_theory_1992, bahouri_fourier_2011} for details on the Littlewood--Paley decomposition in the classical case as well as the classical analogue of the results stated in this section. As in the classical case, we let $\varphi$ be a smooth function of the phase space compactly supported in an annulus $\cC = \set{z\in\Rdd : r \leq \n{z} \leq R}$ such that it is a dyadic partition of the unity, i.e for any $z\neq 0$,
	\begin{equation*}
		\sum_{j\in\Z} \varphi_j(z) = 1
	\end{equation*}
	where $\varphi_j(z) = \varphi(2^{-j}\,z)$, so $\widehat{\varphi_j} = 2^{jd}\,\widehat{\varphi}(2^j\,z)$. To define the quantum analogue of the dyadic block $\triangle_j$ for an operator $\op$, we can use the semiclassical convolution~\eqref{eq:semiclassical_convolution_def}. It leads us to define 
	\begin{equation}
		\triangle_j\op = \widehat{\varphi_j} \star \op = \op_{\widehat{\varphi}_j * f_{\op}} = \op_{\triangle_j f_{\op}},
	\end{equation}
	that is nothing but the Weyl quantization of the dyadic block of the Wigner transform. We can then get a Littlewood--Paley characterization of Besov spaces. We define the homogeneous semi-norm
	\begin{equation*}
		\Nrm{\op}{\tilde{\cB}^s_{p,q}} := \Nrm{2^{js}\Nrm{\triangle_j\op}{\L^p}}{\ell^q}.
	\end{equation*}
	As in the classical case, it gives an equivalent norm to the $\dot{\cB}^s_{p,q}$ norm defined in Section~\ref{sec:besov}. The proof follows the proof given in \cite[Theorem~(2.36) and~(2.37)]{bahouri_fourier_2011}, and we just give the idea on how to adapt the proof by writing the proof of one of the inequalities in the case $s\in(0,1)$.
	\begin{prop}\label{prop:besov_littlewood-paley}
		Let $s\in (0,1)$ and $(p,q)\in [1,\infty]^2$. Then there exists a constant $C>0$ independent of $\hbar$ such that for any $\op\in \dot{\cB}^s_{p,q}$
		\begin{equation}\label{eq:besov_littlewood-paley}
			C^{-1} \Nrm{\op}{\tilde{\cB}^s_{p,q}} \leq \Nrm{\op}{\dot{\cB}^s_{p,q}} \leq C \Nrm{\op}{\tilde{\cB}^s_{p,q}}.
		\end{equation}
	\end{prop}
	
	\begin{proof}[Proof of Proposition~\ref{prop:besov_littlewood-paley}]
		Assume $q<\infty$. The case $q=\infty$ can be treated similarly. Since $s\in(0,1)$, it is sufficient to look at $\Nrm{\op}{\dot{\cB}^s_{p,q,(1)}}$ instead of $\Nrm{\op}{\dot{\cB}^s_{p,q}}$ by Proposition~\ref{prop:Besov_differences_2}. Notice that
		\begin{equation*}
			\sfT_{z}\!\(f \star \op\) = \intdd f(z')\,\sfT_{z'+z} \op\d z' = \intdd f(z'-z)\,\sfT_{z'} \op\d z' = \(\sfT_{z}f\) \star \op.
		\end{equation*}
		Therefore, using the fact that for any $j\in\Z$, $\triangle_j = \sum_{j',\n{j-j'}\leq 1} \triangle_j \triangle_{j'}$, one obtains
		\begin{align*}
			\sfT_{z}\triangle_j\op - \triangle_j\op &= \sum_{j',\n{j-j'}\leq 1} \intdd \(\widehat{\varphi_j}(z'-z) - \widehat{\varphi_j}(z')\) \sfT_{z'} \triangle_{j'}\op \d z'
			\\
			&= \sum_{j',\n{j-j'}\leq 1} \intdd \int_0^1 2^j\, z\cdot \(\nabla\widehat{\varphi}\)_j(z'-\theta z) \d\theta\, \sfT_{z'} \triangle_{j'}\op \d z'
		\end{align*}
		and so since $\Nrm{\(\nabla\widehat{\varphi}\)_j(z'-\theta z)}{L^1_{z'}} = \Nrm{\nabla\widehat{\varphi}}{L^1} =: D_{\varphi}$, it leads to
		\begin{align*}
			\Nrm{\sfT_z\triangle_j\op - \triangle_j\op}{\L^p} &\leq D_{\varphi}\,2^j \n{z} \sum_{j',\n{j-j'}\leq 1} \Nrm{\triangle_{j'}\op}{\L^p} =: D_{\varphi}\,\Nrm{\op}{\tilde{B}^s_{p,q}} 2^{j\(1-s\)} \n{z} d_{j,s}
		\end{align*}
		where $\Nrm{d_{j,s}}{\ell^q} \leq 3^{1/q'}$. On the other hand, by the triangle inequality
		\begin{equation*}
			\Nrm{\sfT_{z}\triangle_j\op - \triangle_j\op}{\L^p} \leq 2\Nrm{\triangle_j\op}{\L^p} \leq 2^{-js} \Nrm{\op}{\tilde{B}^s_{p,q}} \tilde{d}_{j,s}
		\end{equation*}
		where $\Nrm{\tilde{d}_{j,s}}{\ell^q} \leq 2$. Hence, summing over $j$ and using the two above inequalities we deduce
		\begin{equation*}
			\Nrm{\sfT_{z}\op - \op}{\L^p} \leq \(D_{\varphi} \n{z} \sum_{j\leq j_0} 2^{j\(1-s\)} d_{j,s} +  \sum_{j>j_0} 2^{-js} \tilde{d}_{j,s}\) \Nrm{\op}{\tilde{B}^s_{p,q}}.
		\end{equation*}
		The triangle inequality then yields
		\begin{equation}\label{eq:LP_cutoff_1}
			\Nrm{\frac{\Nrm{\sfT_{z}\op - \op}{\L^p}}{\n{z}^{s+2d/q}}}{L^q} \leq \(D_\varphi \,I_1 + I_2\) \Nrm{\op}{\tilde{B}^s_{p,q}},
		\end{equation}
		where
		\begin{align*}
			I_1^q &=  \intdd \frac{1}{\n{z}^{2d-(1-s)q}} \Big(\sum_{j\leq j_0} 2^{j\(1-s\)}\, d_{j,s}\Big)^q \d z
			\\
			I_2^q &= \intdd \frac{1}{\n{z}^{2d+sq}} \Big(\sum_{j>j_0} 2^{-js}\, \tilde{d}_{j,s}\Big)^q \d z.
		\end{align*}
		By H\"older's inequality and the explicit formula for the geometric sum
		\begin{align*}
			\Big(\sum_{j\leq j_0} 2^{j\(1-s\)} d_{j,s}\Big)^q 
			\leq C_{1-s}^{q-1} \, 2^{j_0\(1-s\)\(q-1\)} \sum_{j\leq j_0} 2^{j\(1-s\)} d_{j,s}^q
		\end{align*}
		where $C_{1-s} = \frac{2^{1-s}}{2^{1-s}-1}$. Hence, $I_1$ can be estimated by
		\begin{equation*}
			I_1^q \leq C_{1-s}^{q-1}\intdd\sum_{j\leq j_0} \frac{2^{j_0\(1-s\)\(q-1\)}}{\n{z}^{2d-(1-s)q}}  \, 2^{j\(1-s\)} d_{j,s}^q \d z.
		\end{equation*}
		Taking $j_0$ such that $\frac{\lambda}{\n{z}} \leq 2^{j_0} \leq \frac{2\lambda}{\n{z}}$, one finds by Fubini's theorem
		\begin{equation*}
			I_1^q 
			\leq \frac{\(3C_{1-s}\)^{q-1}\(2\lambda\)^{\(1-s\)q}\,\omega_{2d}}{1-s}.
		\end{equation*}
		Similarly,
		\begin{equation*}
			I_2^q
			\leq \n{C_s}^{q-1} \sum_{j\in\Z} \int_{\n{z}> \lambda \,2^{-j}} \frac{\lambda^{-s\(q-1\)}}{\n{z}^{2d+s}} \, 2^{-js} \tilde{d}_{j,s}^q \d z
			\leq \frac{2^q\n{C_s}^{q-1}\omega_{2d}}{s\,\lambda^{sq}}.
		\end{equation*}
		Combining these two inequalities in Inequality~\eqref{eq:LP_cutoff_1} and optimizing with respect to $\lambda$ yields
		\begin{equation*}
			\Nrm{\frac{\Nrm{\sfT_{z}\op - \op}{\L^p}}{\n{z}^{s+2d/q}}}{L^q}
			\leq \frac{2^{\(1-s\)\(s+1\)}\omega_{2d}^{1/q} D_\varphi^{s} \(3C_{1-s}\)^{s/q'}\(\n{C_s}\)^{(1-s)/q'}}{s^{s+(1-s)/q}\(1-s\)^{1-s+s/q}}  \Nrm{\op}{\tilde{B}^s_{p,q}}
		\end{equation*}
		finishing the proof of the second inequality in Equation~\eqref{eq:besov_littlewood-paley}.
	\end{proof}
	
	From the above proposition, we deduce the classical embedding between Besov spaces with different third index.
	\begin{cor}[Comparison of Quantum Besov spaces]\label{cor:besov_inclusions}
		For any $1\leq r < q \leq \infty$, $\Nrm{\op}{\dot{\cB}^s_{p,q}} \leq C \Nrm{\op}{\dot{\cB}^s_{p,r}}$ for some constant $C > 0$ independent of $\hbar$.
	\end{cor}
	
	From the semiclassical Young inequality, we get a Bernstein-type lemma.
	
	\begin{lem}[Quantum Bernstein Lemma]\label{lem:bernstein}
		For any $1\leq p \leq q \leq \infty$, there exists a constant $C$ independent of $\hbar$ such that for any operator $\op\in\L^p$,
		\begin{equation}
			\Nrm{\Dh^n\triangle_j\op}{\L^q} \leq C\,2^{j\(n+2d\(\frac{1}{p} - \frac{1}{p}\)\)} \Nrm{\op}{\L^p}.
		\end{equation}
	\end{lem}
	
	The above inequalities are a very useful tool. In particular, they allow to prove that the Schatten norms can always be compared to Besov norms in the following way.
	\begin{cor}\label{cor:besov_vs_schatten}
		For any $1\leq p \leq q \leq \infty$ and $r\in[1,\infty]$, there exists $C>0$ such that for any compact operator $\op$
		\begin{equation*}
			C^{-1}\Nrm{\op}{\tilde{\cB}^0_{p,\infty}} \leq \Nrm{\op}{\L^p} \leq C \Nrm{\op}{\tilde{\cB}^0_{p,1}}.
		\end{equation*}
	\end{cor}
	
	\begin{proof}
		The first inequality follows from Bernstein lemma with $p-q=n=0$. To get the second inequality, write $\op = \sum_{j\in \N} \triangle_j\op$, so that
		\begin{equation*}
			\Nrm{\op}{\L^p} \leq \sum_{j\in \N} \Nrm{\triangle_j\op}{\L^p} = \Nrm{\op}{\tilde{\cB}^0_{p,1}}
		\end{equation*}
		follows by the triangle inequality for Schatten norms.
	\end{proof}
	
	From Bernstein's Lemma~\ref{lem:bernstein}, we also obtain the analogue of the Sobolev embedding for Besov spaces.
	\begin{cor}[Sobolev inequalities for quantum Besov spaces]\label{cor:besov_Sobolev_ineq}
		For any $1\leq p \leq q \leq \infty$, $r\in[1,\infty]$ and $(s_0,s_1)\in (0,1)^2$ such that $s_1-s_0 = 2d\(\frac{1}{p}-\frac{1}{q}\)$, there exists $C>0$ such that for any compact operator $\op$
		\begin{equation*}
			 \Nrm{\op}{\dot{\cB}^{s_0}_{q,r}} \leq C \Nrm{\op}{\dot{\cB}^{s_1}_{p,r}}.
		\end{equation*}
	\end{cor}
	
\subsection{Equivalence of norms}

	Recall that a sequence of bounded operators $(A_n)_{n\in\N}$ is said to converge strongly to an operator $A \in \L^\infty$ if and only if
	\begin{equation*}
		\forall \varphi\in L^2(\Rd), A_n\varphi \underset{n\to\infty}{\to} A\varphi.
	\end{equation*}
	It will be denoted by $A_n\overset{\mathrm{s}}{\to} A$. We start by recalling a useful convergence lemma.
	\begin{lem}[Gr\"umm~\cite{grumm_two_1973}]\label{lem:grumm_convergence}
		Assume $A_n\overset{\mathrm{s}}{\to} A\in\L^\infty$ and let $p\in[1,\infty]$ and $B_n$ a sequence of operators of $\L^p$ such that $B_n\to B$ in $\L^p$ for some $B\in \L^p$. Then
		\begin{equation*}
			A_n\,B_n \to A\,B  \text{ in } \L^p.
		\end{equation*}
	\end{lem}
	
	As an application, we deduce the following convergence result.
	\begin{lem}\label{lem:convergence_gradients}
		Let $p\in[1,\infty]$, $\Dh$ denotes either $\Dhx$, $\Dhv$ or any commutator $\com{C,\cdot}$ with an operator $C$ and $\op_n$ be a sequence of operators in $\cW^{1,p}$ such that $\op_n\to\op$ in $\cW^{1,p}$. Assume $A_n\overset{\mathrm{s}}{\to} 1 = \Id_{L^2(\Rd)}$ is a sequence of bounded operators such that $\Dh A_n \overset{\mathrm{s}}{\to} 0$. Then
		\begin{align*}
			&&A_n\,\op_n\, A_n &\to \op &&\text{ in } \L^p&&
			\\
			&&\Dh\!\(A_n\,\op_n\, A_n\) &\to \Dh\op &&\text{ in } \L^p.&&
		\end{align*}
	\end{lem}
	
	\begin{proof}
		The first limit follows by writing $A_n\,\op_n\, A_n - \op = A_n\,\op_n\, A_n - A_n\,\op + A_n\,\op - \op$ so that by the triangle inequality for Schatten norms
		\begin{equation*}
			\Nrm{A_n\,\op_n\, A_n - \op}{\L^p} \leq \big(\sup_{n\in\N}\Nrm{A_n}{\L^\infty}\big) \Nrm{\op_n\, A_n - \op}{\L^p} + \Nrm{A_n\,\op_n - \op}{\L^p}
		\end{equation*}
		and one concludes using Gr\"umm's Lemma~\ref{lem:grumm_convergence}. For the second limit, one proceeds similarly by writing
		\begin{equation*}
			\Dh\!\(A_n\,\op_n\, A_n\) - \Dh\op = \(\Dh A_n\)\op_n\, A_n + A_n \,\op_n \(\Dh A_n\) +  A_n\(\Dh\op_n\) A_n - \Dh\op
		\end{equation*}
		The first two terms of the right-hand side converge to $0$ by Gr\"umm's Lemma. The difference of the two last terms converges to $0$ by the first part of the proof.
	\end{proof}
	
	We are now ready to prove Proposition~\ref{prop:commut_exponential}. Notice moreover that with the notation~\eqref{eq:Bspq_1}, this implies that $\Nrm{\op}{\dot{\cW}^{1,p}}$ is uniformly-in-$\hbar$ equivalent to $\Nrm{\op}{\dot{\cB}^1_{p,\infty,(1)}}$.
	\begin{proof}[Proof of Proposition~\ref{prop:commut_exponential}]
		For $n\in\N\setminus\{0\}$, let $A_n := \(1+\frac{1}{n}\n{x}^2\)^{-1}$. It follows from the dominated convergence theorem that $A_n\overset{\mathrm{s}}{\to} 1$, while by the definition~\eqref{eq:quantum_gradients} of $\Dhvj = \frac{1}{i\hbar}\com{x_{\jj},\cdot}$ where $\jj\in\set{1,\dots,d}$, it follows that $\Dhvj A_n=0$. Moreover, $\op_n := A_n\,\op\,A_n$ satisfies
		\begin{equation}\label{eq:second_gradient_bound}
			\Nrm{\Dhvj^2\op_n}{\L^p} \leq \frac{4}{\hbar^2} \Nrm{\op_n \n{x_\jj}^2}{\L^p} \leq \frac{4\,n}{\hbar^2} \Nrm{\op}{\L^p}
		\end{equation}
		where we used \cite[Inequality~(56)]{lafleche_strong_2023} which tells that for self-adjoint operators $A$ and $B$, $\Nrm{BAB}{\L^p} \leq \Nrm{AB^2}{\L^p}$ and the triangle inequality. Now, by a second order Taylor expansion of the Wigner transform of $\op_n$, and then taking the Weyl quantization, it follows that for any $\xi_\jj = (0,\xi_\jj)\in\Rdd$,
		\begin{equation*}
			\sfT_{\xi_\jj}\op_n - \op_n = \xi_\jj\, \Dhvj{\op_n} + \xi_\jj^2 \int_0^1 \(1-\theta\) \sfT_{\theta \xi_\jj} \Dhvj^2{\op_n} \d \theta.
		\end{equation*}
		Dividing the above equation by $\xi_\jj$ whenever it is nonzero, it follows that
		\begin{equation*}
			\Nrm{\Dhvj{\op_n}}{\L^p} \leq \sup_{\xi_\jj\neq 0}\(\frac{\Nrm{\sfT_{\xi_\jj}\op_n - \op_n}{\L^p}}{\n{\xi_\jj}}\) + \frac{\n{\xi_\jj}}{2} \Nrm{\Dhvj^2{\op_n}}{\L^p}.
		\end{equation*}
		Since $\Nrm{\Dhvj^2\op_n}{\L^p}$ is bounded uniformly with respect to $\xi_\jj$ by Inequality~\eqref{eq:second_gradient_bound}, letting $\xi_j\to 0$ yields
		\begin{equation*}
			\Nrm{\Dhvj{\op_n}}{\L^p} \leq \sup_{\xi_\jj\neq 0}\(\frac{\Nrm{\sfT_{\xi_\jj}\op_n - \op_n}{\L^p}}{\n{\xi_\jj}}\).
		\end{equation*}
		On the other hand, the reversed inequality is true by the first order Taylor formula~\eqref{eq:fundamental_theorem}. Hence
		\begin{equation*}
			\Nrm{\Dhvj{\op_n}}{\L^p} = \sup_{\xi_\jj\neq 0}\(\frac{\Nrm{\sfT_{\xi_\jj}\op_n - \op_n}{\L^p}}{\n{\xi_\jj}}\).
		\end{equation*}
		When $n\to\infty$, the left-hand side converges to $\Nrm{\Dhvj{\op}}{\L^p}$ by Lemma~\ref{lem:convergence_gradients}. On the other hand, by the unitarity of the operator of multiplication by $x\mapsto e^{i\,\xi_\jj\,x_\jj/\hbar}$, it holds
		\begin{equation*}
			\n{\sfT_{\xi_\jj}\op-\op} = \n{e^{i\,\xi_\jj\,x_\jj/\hbar} \op\, e^{i\,\xi_\jj\,x_\jj/\hbar} - \op} = \n{e^{i\,\xi_\jj\,x_\jj/\hbar} \op - \op\, e^{i\,\xi_\jj\,x_\jj/\hbar}} = \n{\com{e^{i\,\xi_\jj\,x_\jj/\hbar}, \op}}.
		\end{equation*}
		Hence
		\begin{equation*}
			\frac{\Nrm{\sfT_{\xi_\jj}\op_n - \op_n}{\L^p}}{\n{\xi_\jj}} \underset{n\to\infty}{\to} \frac{\Nrm{\sfT_{\xi_\jj}\op - \op}{\L^p}}{\n{\xi_\jj}}
		\end{equation*}
		by Lemma~\ref{lem:convergence_gradients}. Verifying that this convergence is uniform with respect to $\xi_\jj$, it follows that
		\begin{equation*}
			\Nrm{\Dhvj{\op}}{\L^p} = \sup_{\xi_\jj\neq 0}\(\frac{\Nrm{\sfT_{\xi_\jj}\op - \op}{\L^p}}{\n{\xi_\jj}}\) = \sup_{\xi_\jj\neq 0}\(\frac{\Nrm{\com{e^{i\,\xi_\jj\,x_\jj/\hbar}, \op}}{\L^p}}{\n{\xi_\jj}}\).
		\end{equation*}
		which proves Equation~\eqref{eq:commut_exponential_j}. Inequalities~\eqref{eq:commut_exponential} then follows on the one side from the first order Taylor expansion Formula~\eqref{eq:fundamental_theorem} and on the other side from the fact that
		\begin{equation*}
			\Nrm{\Dhv\op}{\L^p}^2 \leq \sum_{\jj=1}^d \Nrm{\Dhvj\op}{\L^p}^2 \leq d\, \sup_{z\in\Rdd\setminus\{0\}}\(\frac{\Nrm{\sfT_z\op - \op}{\L^p}}{\n{z}}\)^2
		\end{equation*}
		where the first inequality is proved in appendix (Proposition~\ref{prop:schatten_norm_vec_valued_ineq}).
	\end{proof}
	
\bigskip
\appendix
\section{Vector valued operators}\label{appendix:vec_operators}

	Let $A = (A_1,\dots,A_n)$ and $B = (B_1,\dots,B_n)$ be vector valued operators, i.e. for any $k\in\set{1,\dots,n}$, $A_k$ is an operator acting on $L^2(\Rd)$. Then we denote by $A\cdot B$ the operator $A\cdot B := \sum_{k=1}^n A_k B_k$ and by $\n{A}$ the square root of the positive operator
	\begin{equation*}
		\n{A}^2 = A^*\cdot A = \sum_{k=1}^n \n{A_k}^2.
	\end{equation*}
	More generally, we define $\n{A}_{\ell^r}^r = \sum_{k=1}^n \n{A_k}^r$. We warn the reader about the fact that similarly as the case $r=2$, these quantities do not satisfy in general the triangle inequality. We will use the same notation $\ell^r$ to denote $\Nrm{a_k}{\ell^r}^r = \sum |a_k|^r$ for sequences $(a_k)_{k\geq 1}$, in which case these are indeed norms. The Schatten norm of a vector valued operator is given by $\Nrm{A}{\fS^p} = \Nrm{A}{p} = \(\Tr{\n{A}^p}\)^\frac{1}{p}$. In the case of vector valued operators, we define the Wigner transform component by component, i.e. $f_{A} = (f_{A_1},\dots, f_{A_d})$. One easily verifies that
		\begin{equation*}
			\Nrm{f_A}{L^2} = \Nrm{A}{\L^2}.
		\end{equation*}
	
	\begin{prop}
		Let $x\in\R^n$. Then for any $p\in[0,2]$ $\n{x\cdot A}^p \leq \n{x}^p\n{A}^p$ and for any $p\in[0,\infty]$, $\Nrm{x\cdot A}{p} \leq \n{x}\Nrm{A}{p}$.
	\end{prop}
	
	\begin{proof}
		From the identity
		\begin{equation*}
			\n{x}^2\n{A}^2 - \n{x\cdot A}^2 = \n{x}^{-2}\n{\n{x}^2A - \(x\cdot A\) x}^2
		\end{equation*}
		for any $x\neq 0$, we deduce that $\n{x\cdot A}^2 \leq \n{x}^2\n{A}^2$, and so we deduce the first claimed inequality from the fact that the function $t\to t^a$ is operator monotone for any $a\in[0,1]$ by the L\"owner--Heinz theorem. Taking the trace yields the second inequality for $p\leq 2$. Taking the operator norm gives the case $p=\infty$. The result for $p\in(2,\infty)$ follows by noticing that $f: t\mapsto \n{t}^k t$ is increasing on $\R$ for any $k\geq 0$, so  $A\mapsto \Tr{f(A)}$ is increasing (see e.g.~\cite{carlen_trace_2010}).
	\end{proof}
	
	\begin{prop}\label{prop:schatten_norm_vec_valued_ineq}
		Let $p\in[2,\infty]$, then
		\begin{equation*}
			\(\sum_{k=1}^n\Nrm{A_k}{p}^p\)^{1/p} \leq \Nrm{A}{p} \leq \(\sum_{k=1}^n\Nrm{A_k}{p}^2\)^{1/2}
		\end{equation*}
		while if $p\in[1,2]$, then
		\begin{equation*}
			n^{\frac{1}{2}-\frac{1}{p}}\(\sum_{k=1}^n\Nrm{A_k}{p}^p\)^{1/p} \leq  \Nrm{A}{p} \leq \(\sum_{k=1}^n\Nrm{A_k}{p}^p\)^{1/p}.
		\end{equation*}
	\end{prop}
	
	The above proposition follows by taking $r=2$ in the following inequalities.
	
	\begin{prop}
		Let $(r,p)\in[1,\infty]^2$. Then
		\begin{align*}
			\Nrm{A}{\ell^r(\fS^r)} &\leq \Nrm{A}{\fS^p(\ell^r)} \leq \Nrm{A}{\ell^r(\fS^p)} &&\text{ if } p\geq r
			\\
			n^{\frac{1}{r}-\frac{1}{p}} \Nrm{A}{\ell^p(\fS^p)} &\leq  \Nrm{A}{\fS^p(\ell^r)} \leq \Nrm{A}{\ell^p(\fS^p)} &&\text{ if } p\leq r.
		\end{align*}
	\end{prop}

	\begin{proof}
		If $p\geq r$, then by the triangle inequality for Schatten norms,
		\begin{equation*}
			\Nrm{\n{A}_{\ell^r}}{p}^r = \Nrm{\sum_{k=1}^n \n{A_k}^r}{p/r} \leq \sum_{k=1}^n\Nrm{\n{A_k}^r}{p/r} = \sum_{k=1}^n\Nrm{A_k}{p}^r.
		\end{equation*}
		Conversely, by \cite[Theorem~1.22]{simon_trace_2005}, it holds
		\begin{equation*}
			\sum_{k=1}^n\Nrm{A_k}{p}^p = \sum_{k=1}^n\Nrm{\n{A_k}^r}{p/r}^{p/r} \leq \Nrm{\sum_{k=1}^n\n{A_k}^r}{p/r}^{p/r} = \Nrm{\n{A}_{\ell^r}}{p}^p.
		\end{equation*}
		Now let $p\in [1,r]$. Then, $t\mapsto t^{p/r}$ is operator concave on positive operators, which yields
		\begin{equation*}
			\sum_{k=1}^n \n{A_k}^{p} \leq n^{1-\frac{p}{r}} \(\sum_{k=1}^n \n{A_k}^{r}\)^\frac{p}{r}
		\end{equation*}
		and so taking the trace,
		\begin{equation*}
			\sum_{k=1}^n \Nrm{A_k}{p}^p \leq n^{1-\frac{p}{r}} \Nrm{\n{A}_{\ell^r}}{p}^p.
		\end{equation*}
		Finally, assume $n=2$. Then since $\n{A_1}^r+\n{A_2}^r \geq \n{A_1}^r$, and $t\mapsto t^{p/r}$ is operator monotone, we deduce that $\(\n{A_1}^r+\n{A_2}^r\)^\frac{p}{r} - \n{A_1}^p \geq 0$. Since by the Birman--Koplienko--Solomyak--Ando inequality~\cite{birman_estimates_1975, ando_comparison_1988}
		\begin{equation*}
			\Tr{\n{\(\n{A_1}^r+\n{A_2}^r\)^\frac{p}{r} - \(\n{A_1}^r\)^\frac{p}{r}}} \leq \Tr{\(\n{A_2}^r\)^\frac{p}{r}},
		\end{equation*}
		we deduce that $\Tr{\(\n{A_1}^r+\n{A_2}^r\)^\frac{p}{r}} \leq \Tr{\n{A_1}^p+\n{A_2}^p}$. By induction, this generalizes to the case $n\geq 2$ leading to $\Nrm{\n{A}_{\ell^r}}{p}^p \leq \sum_{k=1}^n \Nrm{A_k}{p}^p$.
	\end{proof}
	
	We recall the Clarkson--McCarthy inequalities~\cite{mccarthy_c_p_1967}
	\begin{align*}
		\Nrm{A-B}{p}^{p'} + \Nrm{A+B}{p}^{p'} &\leq 2\(\Nrm{A}{p}^p + \Nrm{B}{p}^p\)^\frac{p'}{p} && \text{ if } p\in(1,2]
		\\
		\Nrm{A-B}{p}^p + \Nrm{A+B}{p}^p &\leq 2\(\Nrm{A}{p}^{p'} + \Nrm{B}{p}^{p'}\)^\frac{p}{p'} && \text{ if } p\in[2,\infty)
	\end{align*}
	More generally, these can be written as follows.
	\begin{prop}\label{prop:Clarkson-McCarthy}
		Let $(p,r)\in[1,\infty]^2$ be such that $r \geq \max(p,p')$. Then defining $M = (A-B, A+B)$ and $N = (A,B)$, it holds
		\begin{align}\label{eq:Clarkson-McCarthy}
			\Nrm{M}{\ell^r(\L^p)} &\leq 2^{1/r} \Nrm{N}{\ell^{r'}(\L^p)}
			\\\label{eq:Clarkson-McCarthy_2}
			2^{1/r'} \Nrm{N}{\ell^{r}(\L^p)} &\leq \Nrm{M}{\ell^{r'}(\L^p)}.
		\end{align}
	\end{prop}
	
	\begin{proof}
		Following~\cite{boas_uniformly_1940}, the result can be proved by interpolation. Notice first that $\n{M}_{\ell^2}^2 = 2\n{N}_{\ell^2}^2$, to deduce that
		\begin{equation*}
			\Nrm{M}{\ell^2(\L^2)} = \sqrt{2} \Nrm{N}{\ell^2(\L^2)}.
		\end{equation*}
		On the other hand, for any $p_1\in[1,\infty]$, by the triangle inequality for Schatten norms,
		\begin{equation*}
			\Nrm{M}{\ell^\infty(\L^{p_1})} \leq \Nrm{N}{\ell^1(\L^{p_1})}.
		\end{equation*}
		Therefore, since $r\geq 2$ and $p\in[r',r]$, the Inequality~\eqref{eq:Clarkson-McCarthy} follows using the fact that the spaces $\ell^r(\L^p)$ are an exact family of complex interpolation (see~\cite[Satz~8]{pietsch_interpolationstheorie_1968} and~\cite[1.18.1]{triebel_interpolation_1978}). Inequality~\eqref{eq:Clarkson-McCarthy_2} then follows from Equation~\eqref{eq:Clarkson-McCarthy} just by replacing $(A, B)$ by $(\frac{A+B}{2}, \frac{A-B}{2})$.
	\end{proof}

\section{The quantum fractional Laplacian}\label{appendix:frac_lap}

	As indicated in Section~\ref{sec:quantum_Sobolev_spaces}, a possible definition for the quantum fractional Laplacian is through the Wigner transform and Weyl quantization by the Formula
	\begin{equation}\label{eq:lap_frac_def_2}
		(-\DDh)^s\op := \op_{(-\Delta_z)^s f_{\op}}.
	\end{equation}
	In this section, we justify the identities stated in Section~\ref{sec:quantum_Sobolev_spaces} about this operator and state some other identities. Notice first that when $s=1$, it follows directly by linearity of the Weyl quantization that
	\begin{equation*}
		(-\DDh)\op = \op_{(-\Dx\cdot\Dx) f_{\op}} + \op_{(-\Dv\cdot\Dv) f_{\op}},
	\end{equation*}
	Therefore, by Equation~\eqref{eq:weyl_quantization_gradients}, it follows that
	\begin{equation*}
		\DDh\op := -(-\DDh)\op = \Dhx\cdot\Dhx\op + \Dhv\cdot\Dhv\op.
	\end{equation*}
	In particular, this can be written as $\DDh\op = \DDh_x\op + \DDh_\xi\op$ with $\DDh_x := \Dhx\cdot\Dhx$ and $\DDh_\xi := \Dhv\cdot\Dhv$. Moreover, since the integral kernel of $\Dhv\op$ is given by $\frac{1}{i\hbar} \(x-y\) \op(x,y)$, one deduces that the integral kernel of $\DDh_\xi\op$ is given by
	\begin{equation*}
		\DDh_\xi\op(x,y) = -\,\frac{\n{x-y}^2}{\hbar^2}\, \op(x,y).
	\end{equation*}
	More generally, defining $(-\DDh_\xi)^\frac{s}{2}\op := \op_{(-\Delta_\xi)^\frac{s}{2} f_{\op}}$, the following result holds.
	\begin{prop}
		Let $s\in[0,1]$ and $\op\in\L^2$ be such that $(-\DDh_\xi)^s\op\in\L^2$. Then the integral kernel of $(-\DDh_\xi)^s\op$ is given by
		\begin{equation*}
			(-\DDh_\xi)^\frac{s}{2}\op(x,y) = \frac{\n{x-y}^s}{\hbar^s}\, \op(x,y).
		\end{equation*}
	\end{prop}
	
	\begin{proof}
		Recalling that $f_{\op}(x,\xi) = \cF_y\!\(\op(x+\tfrac{y}{2},x-\tfrac{y}{2})\)\!\(\frac{\xi}{h}\)$, it follows from the classical Fourier definition of the fractional Laplacian and its scaling properties that
		\begin{align*}
			\((-\Delta_\xi)^\frac{s}{2} f_{\op}\)(x,h\xi) &= \frac{1}{h^s} \(-\Delta_\xi\)^\frac{s}{2}\!\cF_y\!\(\op(x+\tfrac{y}{2},x-\tfrac{y}{2})\)\!\(\xi\)
			\\
			&= \frac{1}{h^s} \cF_y\!\(\n{2\pi y}^s\op(x+\tfrac{y}{2},x-\tfrac{y}{2})\)\!\(\xi\)
		\end{align*}
		and so since $\n{y} = \n{x+\tfrac{y}{2}-\(x-\tfrac{y}{2}\)}$, using again the Fourier definition of the Wigner transform yields
		\begin{equation*}
			\((-\Delta_\xi)^\frac{s}{2} f_{\op}\)(x,h\xi) = \frac{1}{\hbar^s}\, f_{\n{x-y}^s\op(x,y)}(x,h\xi).
		\end{equation*}
		Replacing $h\xi$ by $\xi$ and taking the Weyl quantization on both sides finishes the proof.
	\end{proof}
	
	Let us now give a proof of Equation~\eqref{eq:lap_frac_integral} in the case when the integral can be written as a Bochner integral.
	\begin{prop}
		Let $s\in (0,1)$ and $\op \in\cB^s_{2,1}$. Then
	\begin{equation}\label{eq:lap_frac_integral_2}
		(-\DDh)^\frac{s}{2}\op = c_{2d,s} \intdd \frac{\sfT_{z}\op - \op}{\n{z}^{2d+s}} \d z.
	\end{equation}
	\end{prop}
	
	\begin{proof}
		Notice first that since the Wigner transform is an isometry from $\L^2$ to $L^2(\Rdd)$ and $\op\in\cB^s_{2,1}$, it holds
		\begin{equation*}
			\intdd \frac{\Nrm{f_{\sfT_z\op - \op}}{L^2(\Rdd)}}{\n{z}^{2d+s}} \d z = \intdd \frac{\Nrm{\sfT_z\op - \op}{\L^2}}{\n{z}^{2d+s}} \d z = \Nrm{\op}{\dot\cB^s_{2,1}} < \infty.
		\end{equation*}
		It follows from Formula~\eqref{eq:Weyl_translation} that $f_{\sfT_z\op} = f_{\op}(\cdot-z)$, hence by linearity of the Wigner transform, it yields
		\begin{equation*}
			\intdd \Nrm{\frac{f_{\op}(\cdot-z) - f_{\op}}{\n{z}^{2d+s}}}{L^2(\Rdd)} \d z < \infty.
		\end{equation*}
		Therefore, by the integral definition~\eqref{eq:lap_frac_integral_cl}, the fractional Laplacian of $f_{\op}$ is a well-defined Bochner integral given by
		\begin{equation*}
			(-\Delta)^\frac{s}{2} f_{\op} = c_{2d,s} \intdd \frac{f_{\op}(\cdot-z)-f_{\op}}{\n{z}^{2d+s}} \d z.
		\end{equation*}
		Since the Weyl quantization is an isometry from $L^2(\Rdd)$ to $\L^2$, it follows that
		\begin{equation*}
			(-\DDh)^s\op = \op_{(-\Delta_z)^s f_{\op}} = c_{2d,s} \intdd \frac{\op_{f_{\op}(\cdot-z)}-\op_{f_{\op}}}{\n{z}^{2d+s}} \d z,
		\end{equation*}
		and so since $\op_{f_{\op}(\cdot-z)} = \sfT_z\op_{f_{\op}}$ by Formula~\eqref{eq:Weyl_translation} and $\op_{f_{\op}} = \op$, this leads to the claimed formula.
	\end{proof}

\medskip
\paragraph{\bf Acknowledgment.} This project has received funding from the European Research Council (ERC) under the European Union’s Horizon 2020 research and innovation program (grant agreement No 865711).


\renewcommand{\bibname}{\centerline{Bibliography}}
\bibliographystyle{abbrv} 
\bibliography{Vlasov}

\begin{thebibliography}{10}

\bibitem{adams_sobolev_2003}
R.~A. Adams and J.~J.~F. Fournier.
\newblock {\em Sobolev {{Spaces}}}.
\newblock {Academic Press}, {Amsterdam; Boston}, 2nd edition, 2003.

\bibitem{ando_comparison_1988}
T.~Ando.
\newblock Comparison of {{Norms}} {$|||f(A) - f(B)|||$} and {$|||f(|A -
  B|)|||$}.
\newblock {\em Mathematische Zeitschrift}, 197(3):403--410, 1988.

\bibitem{aubin_problemes_1976}
T.~Aubin.
\newblock {Probl{\`e}mes isop{\'e}rim{\'e}triques et espaces de Sobolev}.
\newblock {\em Journal of Differential Geometry}, 11(4):573--598, Jan. 1976.

\bibitem{bahouri_fourier_2011}
H.~Bahouri, J.-Y. Chemin, and R.~Danchin.
\newblock {\em Fourier {{Analysis}} and {{Nonlinear Partial Differential
  Equations}}}, volume 343 of {\em Grundlehren Der Mathematischen
  {{Wissenschaften}}}.
\newblock {Springer}, {Berlin, Heidelberg}, Jan. 2011.

\bibitem{bauer_self-adjointness_2023}
W.~Bauer, L.~{van Luijk}, A.~Stottmeister, and R.~F. Werner.
\newblock Self-adjointness of {{Toeplitz}} operators on the
  {{Segal}}{\textendash}{{Bargmann}} space.
\newblock {\em Journal of Functional Analysis}, 284(4):109778, Feb. 2023.

\bibitem{benedikter_effective_2022}
N.~Benedikter.
\newblock Effective {{Dynamics}} of {{Interacting Fermions}} from
  {{Semiclassical Theory}} to the {{Random Phase Approximation}}.
\newblock {\em Journal of Mathematical Physics}, 63(8):081101, Aug. 2022.

\bibitem{benedikter_hartree_2016}
N.~Benedikter, M.~Porta, C.~Saffirio, and B.~Schlein.
\newblock From the {{Hartree Dynamics}} to the {{Vlasov Equation}}.
\newblock {\em Archive for Rational Mechanics and Analysis}, 221(1):273--334,
  July 2016.

\bibitem{benedikter_mean-field_2014}
N.~Benedikter, M.~Porta, and B.~Schlein.
\newblock Mean-field {{Evolution}} of {{Fermionic Systems}}.
\newblock {\em Communications in Mathematical Physics}, 331(3):1087--1131, Nov.
  2014.

\bibitem{benedikter_mean-field_2016-1}
N.~Benedikter, M.~Porta, and B.~Schlein.
\newblock Mean-{{Field Regime}} for {{Fermionic Systems}}.
\newblock In {\em Effective {{Evolution Equations}} from {{Quantum Dynamics}}},
  pages 57--78. {Springer}, 2016.

\bibitem{bergh_interpolation_1976}
J.~Bergh and J.~L{\"o}fstr{\"o}m.
\newblock {\em Interpolation Spaces. {{An}} Introduction}, volume 223 of {\em
  Grundlehren Der {{Mathematischen Wissenschaften}}}.
\newblock {Springer Berlin Heidelberg}, {Berlin, Heidelberg}, 1976.

\bibitem{birman_estimates_1975}
M.~S. Birman, L.~S. Koplienko, and M.~Z. Solomyak.
\newblock Estimates for the spectrum of the difference between fractional
  powers of two self-adjoint operators.
\newblock {\em Soviet Mathematics}, 19(3):1--6, 1975.

\bibitem{boas_uniformly_1940}
R.~P. Boas.
\newblock Some uniformly convex spaces.
\newblock {\em Bulletin of the American Mathematical Society}, 46(4):304--311,
  Apr. 1940.

\bibitem{boulkhemair_l2_1999}
A.~Boulkhemair.
\newblock {$L^2$} estimates for {{Weyl Quantization}}.
\newblock {\em Journal of Functional Analysis}, 165(1):173--204, June 1999.

\bibitem{bourgain_another_2001}
J.~Bourgain, H.~Brezis, and P.~Mironescu.
\newblock Another look at {{Sobolev}} spaces.
\newblock In {\em Optimal {{Control}} and {{Partial Differential Equations}}},
  pages 439--455, 2001.

\bibitem{brezis_how_2002}
H.~Brezis.
\newblock How to recognize constant functions. {{Connections}} with {{Sobolev}}
  spaces.
\newblock {\em Russian Mathematical Surveys}, 57(4):693--708, Aug. 2002.

\bibitem{brezzi_three-dimensional_1991}
F.~Brezzi and P.~A. Markowich.
\newblock {The Three-Dimensional Wigner-Poisson Problem: Existence, Uniqueness
  and Approximation}.
\newblock {\em Mathematical Methods in the Applied Sciences}, 14(1):35--61,
  Jan. 1991.

\bibitem{calderon_lebesgue_1961}
A.~P. Calder{\'o}n.
\newblock Lebesgue spaces of differentiable functions.
\newblock In {\em Proc. {{Sympos}}. {{Pure Math}}}, volume~4, pages 33--49,
  {Providence, R.I.}, 1961. {American Mathematical Society}.

\bibitem{calderon_intermediate_1964}
A.~P. Calder{\'o}n.
\newblock Intermediate {{Spaces}} and {{Interpolation}}, the {{Complex
  Method}}.
\newblock {\em Studia Mathematica}, 24(2):113--190, 1964.

\bibitem{cardenas_effective_2023}
E.~C{\'a}rdenas, J.~K. Miller, and N.~Pavlovi{\'c}.
\newblock On the effective dynamics of {{Bose-Fermi}} mixtures.
\newblock {\em arXiv:2309.04638}, pages 1--55, Sept. 2023.

\bibitem{carlen_trace_2010}
E.~Carlen.
\newblock Trace inequalities and quantum entropy: An introductory course.
\newblock In {\em Contemporary {{Mathematics}}}, volume 529, pages 73--140.
  {American Mathematical Society}, {Providence, Rhode Island}, 2010.

\bibitem{caspers_schur_2015}
M.~Caspers and M.~{de la Salle}.
\newblock Schur and {{Fourier}} multipliers of an amenable group acting on
  non-commutative {{Lp-spaces}}.
\newblock {\em Transactions of the American Mathematical Society},
  367(10):6997--7013, Mar. 2015.

\bibitem{castella_$l^2$_1997}
F.~Castella.
\newblock {$L^2$} solutions to the {{Schr{\"o}dinger}}{\textendash}{{Poisson
  System}}: {{Existence}}, {{Uniqueness}}, {{Time Behaviour}}, and {{Smoothing
  Effects}}.
\newblock {\em Mathematical Models and Methods in Applied Sciences},
  07(08):1051--1083, Dec. 1997.

\bibitem{chadam_time-dependent_1976}
J.~M. Chadam.
\newblock The {{Time-Dependent Hartree}}{\textendash}{{Fock Equations}} with
  {{Coulomb Two-Body Interaction}}.
\newblock {\em Communications in Mathematical Physics}, 46(2):99--104, June
  1976.

\bibitem{chong_many-body_2021}
J.~J. Chong, L.~Lafleche, and C.~Saffirio.
\newblock From {{Many-Body Quantum Dynamics}} to the
  {{Hartree}}{\textendash}{{Fock}} and {{Vlasov Equations}} with {{Singular
  Potentials}}.
\newblock {\em arXiv:2103.10946}, pages 1--74, Mar. 2021.

\bibitem{chong_global--time_2022}
J.~J. Chong, L.~Lafleche, and C.~Saffirio.
\newblock Global-in-time {{Semiclassical Regularity}} for the
  {{Hartree}}{\textendash}{{Fock Equation}}.
\newblock {\em Journal of Mathematical Physics}, 63(8):081904, Aug. 2022.

\bibitem{chong_l2_2023}
J.~J. Chong, L.~Lafleche, and C.~Saffirio.
\newblock On the {$L^2$} {{Rate}} of {{Convergence}} in the {{Limit}} from the
  {{Hartree}} to the {{Vlasov}}{\textendash}{{Poisson Equation}}.
\newblock {\em Journal de l'{\'E}cole polytechnique {\textendash}
  Math{\'e}matiques}, 10:703--726, 2023.

\bibitem{chong_semiclassical_2023}
J.~J. Chong, L.~Lafleche, and C.~Saffirio.
\newblock On the {{Semiclassical Regularity}} of {{Thermal Equilibria}}.
\newblock In {\em Quantum {{Mathematics I}}}, Springer {{INdAM Series}}, pages
  89--105, {Singapore}, Dec. 2023. {Springer Nature}.

\bibitem{conde-alonso_schur_2022}
J.~M. {Conde-Alonso}, A.~M. {Gonz{\'a}lez-P{\'e}rez}, J.~Parcet, and
  E.~Tablate.
\newblock Schur multipliers in {{Schatten}}{\textendash}von {{Neumann}}
  classes.
\newblock {\em arXiv:2201.05511}, pages 1--20, Jan. 2022.

\bibitem{de_palma_quantum_2021-1}
G.~De~Palma, M.~Marvian, D.~Trevisan, and S.~Lloyd.
\newblock The {{Quantum Wasserstein Distance}} of {{Order}} 1.
\newblock {\em IEEE Transactions on Information Theory}, 67(10):6627--6643,
  Oct. 2021.

\bibitem{erdos_derivation_2001}
L.~Erd{\"o}s and H.-T. Yau.
\newblock Derivation of the {{Nonlinear Schr{\"o}dinger Equation}} from a
  {{Many Body Coulomb System}}.
\newblock {\em Advances in Theoretical and Mathematical Physics},
  5(6):1169--1205, 2001.

\bibitem{fournais_optimal_2020}
S.~Fournais and S.~Mikkelsen.
\newblock An optimal semiclassical bound on commutators of spectral projections
  with position and momentum operators.
\newblock {\em Letters in Mathematical Physics}, 110(12):3343--3373, Dec. 2020.

\bibitem{fulsche_correspondence_2020}
R.~Fulsche.
\newblock Correspondence theory on p-{{Fock}} spaces with applications to
  {{Toeplitz}} algebras.
\newblock {\em Journal of Functional Analysis}, 279(7):108661, Oct. 2020.

\bibitem{gagliardo_ulteriori_1959}
E.~Gagliardo.
\newblock Ulteriori propriet{\`a} di alcune classi di funzioni in pi{\`u}
  variabili.
\newblock {\em Ricerche Mat.}, 8:24--51, 1959.

\bibitem{golse_convergence_2021}
F.~Golse, S.~Jin, and T.~Paul.
\newblock On the {{Convergence}} of {{Time Splitting Methods}} for {{Quantum
  Dynamics}} in the {{Semiclassical Regime}}.
\newblock {\em Foundations of Computational Mathematics}, 21(3):613--647, June
  2021.

\bibitem{golse_empirical_2019}
F.~Golse and T.~Paul.
\newblock Empirical {{Measures}} and {{Quantum Mechanics}}: {{Application}} to
  the {{Mean-Field Limit}}.
\newblock {\em Communications in Mathematical Physics}, 369(3):1021--1053, Aug.
  2019.

\bibitem{golse_semiclassical_2021}
F.~Golse and T.~Paul.
\newblock Semiclassical {{Evolution With Low Regularity}}.
\newblock {\em Journal de Math{\'e}matiques Pures et Appliqu{\'e}es},
  151:257--311, July 2021.

\bibitem{grumm_two_1973}
H.~R. Gr{\"u}mm.
\newblock Two theorems about {{Cp}}.
\newblock {\em Reports on Mathematical Physics}, 4(3):211--215, May 1973.

\bibitem{junge_noncommutative_2018}
M.~Junge, T.~Mei, and J.~Parcet.
\newblock Noncommutative {{Riesz}} transforms {\textendash} dimension free
  bounds and {{Fourier}} multipliers.
\newblock {\em Journal of the European Mathematical Society}, 20(3):529--595,
  Feb. 2018.

\bibitem{karlovich_algebras_2007}
Y.~I. Karlovich.
\newblock Algebras of {{Pseudo-differential Operators}} with {{Discontinuous
  Symbols}}.
\newblock In {\em Modern {{Trends}} in {{Pseudo-Differential Operators}}},
  volume 172 of {\em Operator {{Theory}}: {{Advances}} and {{Applications}}},
  pages 207--233. {Birkh{\"a}user}, {Basel}, 2007.

\bibitem{kosaki_matrix_2005}
H.~Kosaki.
\newblock Matrix {{Trace Inequalities Related}} to {{Uncertainty Principle}}.
\newblock {\em International Journal of Mathematics}, 16(6):629--645, July
  2005.

\bibitem{lafleche_propagation_2019}
L.~Lafleche.
\newblock Propagation of {{Moments}} and {{Semiclassical Limit}} from
  {{Hartree}} to {{Vlasov Equation}}.
\newblock {\em Journal of Statistical Physics}, 177(1):20--60, Oct. 2019.

\bibitem{lafleche_optimal_2023}
L.~Lafleche.
\newblock Optimal {{Semiclassical Regularity}} of {{Projection Operators}} and
  {{Strong Weyl Law}}.
\newblock {\em arXiv:2302.04816}, pages 1--20, Feb. 2023.

\bibitem{lafleche_quantum_2023}
L.~Lafleche.
\newblock Quantum {{Optimal Transport}} and {{Weak Topologies}}.
\newblock {\em arXiv:2306.12944}, pages 1--25, June 2023.

\bibitem{lafleche_strong_2023}
L.~Lafleche and C.~Saffirio.
\newblock Strong {{Semiclassical Limits}} from {{Hartree}} and
  {{Hartree}}{\textendash}{{Fock}} to {{Vlasov}}{\textendash}{{Poisson
  Equations}}.
\newblock {\em Analysis \& PDE}, 16(4):891--926, June 2023.

\bibitem{lieb_sharp_1983}
E.~H. Lieb.
\newblock Sharp {{Constants}} in the
  {{Hardy}}{\textendash}{{Littlewood}}{\textendash}{{Sobolev}} and {{Related
  Inequalities}}.
\newblock {\em Annals of Mathematics}, 118(2):349--374, Sept. 1983.

\bibitem{lieb_analysis_2001}
E.~H. Lieb and M.~Loss.
\newblock {\em Analysis}, volume~14 of {\em Graduate {{Studies}} in
  {{Mathematics}}}.
\newblock {American Mathematical Society}, {Providence, RI}, 2nd edition, 2001.

\bibitem{lions_sur_1993}
P.-L. Lions and T.~Paul.
\newblock Sur les mesures de {{Wigner}}.
\newblock {\em Revista Matem{\'a}tica Iberoamericana}, 9(3):553--618, 1993.

\bibitem{luo_heisenberg_2005}
S.~Luo.
\newblock Heisenberg uncertainty relation for mixed states.
\newblock {\em Physical Review A}, 72(4):042110, Oct. 2005.

\bibitem{luo_informational_2004}
S.~Luo and Z.~Zhang.
\newblock An {{Informational Characterization}} of {{Schr{\"o}dinger}}'s
  {{Uncertainty Relations}}.
\newblock {\em Journal of Statistical Physics}, 114(5):1557--1576, Mar. 2004.

\bibitem{marcantoni_dynamics_2023}
S.~Marcantoni, M.~Porta, and J.~Sabin.
\newblock Dynamics of mean-field {{Fermi}} systems with nonzero pairing.
\newblock {\em arXiv:2310.15280}, pages 1--43, Oct. 2023.

\bibitem{mazya_sobolev_2011}
V.~Maz'ya.
\newblock {\em Sobolev {{Spaces}}}, volume 342 of {\em Grundlehren Der
  Mathematischen {{Wissenschaften}}}.
\newblock {Springer Berlin Heidelberg}, {Berlin, Heidelberg}, 2011.

\bibitem{mazya_bourgain_2002}
V.~Maz'ya and T.~Shaposhnikova.
\newblock On the {{Bourgain}}, {{Brezis}}, and {{Mironescu Theorem Concerning
  Limiting Embeddings}} of {{Fractional Sobolev Spaces}}.
\newblock {\em Journal of Functional Analysis}, 195(2):230--238, Nov. 2002.

\bibitem{mccarthy_c_p_1967}
C.~A. McCarthy.
\newblock {$c_p$}.
\newblock {\em Israel Journal of Mathematics}, 5(4):249--271, Oct. 1967.

\bibitem{nirenberg_elliptic_1959}
L.~Nirenberg.
\newblock On {{Elliptic Partial Differential Equations}}.
\newblock {\em Annali della Scuola Normale Superiore di Pisa - Scienze Fisiche
  e Matematiche}, 13(2):115--162, 1959.

\bibitem{pietsch_interpolationstheorie_1968}
A.~Pietsch and H.~Triebel.
\newblock {Interpolationstheorie f{\"u}r Banachideale von beschr{\"a}nkten
  linearen Operatoren}.
\newblock {\em Studia Mathematica}, 31(1):95--109, 1968.

\bibitem{porta_mean_2017}
M.~Porta, S.~Rademacher, C.~Saffirio, and B.~Schlein.
\newblock Mean {{Field Evolution}} of {{Fermions}} with {{Coulomb
  Interaction}}.
\newblock {\em Journal of Statistical Physics}, 166(6):1345--1364, Mar. 2017.

\bibitem{reed_functional_1980}
M.~C. Reed and B.~Simon.
\newblock {\em Functional {{Analysis}}}, volume~1 of {\em Methods of {{Modern
  Mathematical Physics}}}.
\newblock {Academic Press}, {New York}, revised and enlarged edition, Jan.
  1980.

\bibitem{saffirio_semiclassical_2019}
C.~Saffirio.
\newblock Semiclassical {{Limit}} to the {{Vlasov Equation}} with {{Inverse
  Power Law Potentials}}.
\newblock {\em Communications in Mathematical Physics}, 373(2):571--619, Mar.
  2019.

\bibitem{simon_trace_2005}
B.~Simon.
\newblock {\em Trace {{Ideals}} and Their {{Applications}}: {{Second
  Edition}}}, volume 120 of {\em Mathematical {{Surveys}} and {{Monographs}}}.
\newblock {American Mathematical Society}, 2 edition, 2005.

\bibitem{sobolev_theorem_1938}
S.~L. Sobolev.
\newblock On a {{Theorem}} of {{Functional Analysis}}.
\newblock {\em Mat. Sbornik}, 46:471--497, 1938.

\bibitem{stein_singular_1970}
E.~M. Stein.
\newblock {\em Singular {{Integrals}} and {{Differentiability Properties}} of
  {{Functions}}}.
\newblock {Princeton University Press}, {Princeton, N.J.}, 1970.

\bibitem{talenti_best_1976}
G.~Talenti.
\newblock Best {{Constant}} in {{Sobolev Inequality}}.
\newblock {\em Annali di Matematica Pura ed Applicata}, 110(1):353--372, Dec.
  1976.

\bibitem{tartar_introduction_2007}
L.~Tartar.
\newblock {\em An {{Introduction}} to {{Sobolev Spaces}} and {{Interpolation
  Spaces}}}, volume~3 of {\em Lecture {{Notes}} of the {{Unione Matematica
  Italiana}}}.
\newblock {Springer-Verlag Berlin Heidelberg}, {Berlin, Heidelberg}, 1 edition,
  2007.

\bibitem{triebel_interpolation_1978}
H.~Triebel.
\newblock {\em Interpolation {{Theory}}, {{Function Spaces}}, {{Differential
  Operators}}}.
\newblock Number~18 in North-{{Holland Mathematical Library}}. {Elsevier
  Science}, 1978.

\bibitem{triebel_theory_1992}
H.~Triebel.
\newblock {\em Theory of {{Function Spaces II}}}.
\newblock {Springer Basel}, {Basel}, 1992.

\bibitem{werner_quantum_1984}
R.~Werner.
\newblock Quantum harmonic analysis on phase space.
\newblock {\em Journal of Mathematical Physics}, 25(5):1404--1411, May 1984.

\bibitem{wigner_information_1963}
E.~P. Wigner and M.~M. Yanase.
\newblock Information {{Contents}} of {{Distributions}}.
\newblock {\em Proceedings of the National Academy of Sciences of the United
  States of America}, 49(6):910--918, June 1963.

\bibitem{yanagi_wigneryanasedyson_2010}
K.~Yanagi.
\newblock Wigner{\textendash}{{Yanase}}{\textendash}{{Dyson}} skew information
  and uncertainty relation.
\newblock {\em Journal of Physics: Conference Series}, 201:012015, Dec. 2010.

\bibitem{yang_generalized_2022}
M.-C. Yang and C.-F. Qiao.
\newblock Generalized {{Wigner}}{\textendash}{{Yanase Skew Information}} and
  the {{Affiliated Inequality}}.
\newblock {\em Physical Review A}, 106(5):052401, Nov. 2022.

\end{thebibliography}

\end{document}